\documentclass[11pt]{article}

\usepackage{setspace}
\usepackage[margin=1in]{geometry}

\usepackage{natbib}
\bibpunct[, ]{(}{)}{,}{a}{}{,}

\usepackage{amsmath,amssymb,amsthm}
\usepackage{graphicx}
\usepackage{booktabs}
\usepackage{tabularx}
\usepackage{multirow}
\usepackage{float}
\usepackage{makecell}
\usepackage[colorlinks=true, linkcolor=blue, citecolor=blue,
            urlcolor=blue]{hyperref}

\newtheorem{lemma}{Lemma}
\newtheorem{proposition}{Proposition}
\newtheorem{corollary}{Corollary}
\newtheorem{assumption}{Assumption}
\newtheorem{remark}{Remark}

\title{The Security Cost of Intelligence:\\
AI Capability, Cyber Risk, and Deployment Paradox}

\author{Sukwoong Choi\thanks{Assistant Professor, Massry School of Business, University at Albany, SUNY; Research Fellow, MIT FutureTech. Email: \href{mailto:schoi27@albany.edu}{schoi27@albany.edu}.}}

\date{April 24, 2026}

\begin{document}

\maketitle

\begin{abstract}
Firms are deploying more capable AI systems, but organizational controls often
have not kept pace. These systems can generate greater productivity gains, but
high-value uses require broader authority exposure---data access, workflow
integration, and delegated authority---when governance controls have not yet
decoupled capability from authority exposure. We develop an analytical model in
which a firm jointly chooses AI deployment and cybersecurity investment under
this governance-capability gap. The central result shows a deployment paradox:
in high-loss environments, better AI can lead a firm to deploy less when
capability is deployed through broader authority exposure under weak governance.
Optimal deployment also falls below the no-risk benchmark, and this shortfall
widens with breach-loss magnitude and with the authority exposure attached to
more capable systems. Governance investment that reduces breach-loss magnitude
shrinks the paradox region itself, while breach externalities expand the range
of environments in which deployment is socially constrained. Governance
maturity is therefore not merely a constraint on AI adoption. It is a condition
that shapes whether capability improvements translate into productive deployment.
\end{abstract}


\newpage

\section{Introduction}

In many current enterprise deployments, governance maturity has not kept pace
with operational integration. High-value uses of capable AI systems often
require broad data access, workflow integration, and delegated authority. We
use the term \textit{authority exposure} to refer to this bundle of access,
workflow reach, and delegated action rights. When governance controls cannot
limit the authority exposure attached to these systems, more capable
deployments can become harder to contain. In November
2025, Anthropic disclosed a large-scale AI-orchestrated cyber espionage
campaign in which a jailbroken Claude Code agent automated major stages of
attack execution \citep{anthropic2025espionage}. The same organizational
features that make AI productive---deep data access, cross-system integration,
and autonomous decision-making---can also expand the attack surface and
increase breach consequences when governance has not kept pace
\citep{kumar2020, microsoft2024}. Among organizations that reported AI-related
security incidents, 97\% lacked proper AI access controls \citep{ibm2025}. A
2026 survey by the Cloud Security Alliance found that 74\% of organizations
give AI agents more privileges than necessary, 79\% report agents creating
hard-to-monitor access pathways, and only 22\% apply access-control frameworks
consistently \citep{csa2026}.

These findings do not mean that AI capability and operational authority are
inherently inseparable. In principle, the two can be decoupled: organizations
can grant AI systems access only to the data and decision rights each task
requires, rather than broad default permissions---an approach known in
cybersecurity practice as Zero Trust architecture and least-privilege design
\citep{nist2020zerotrust}. Yet they suggest that governance maturity has not
kept pace with capability advancement. We call this misalignment the
\textit{governance-capability gap}. By governance, we mean the organizational
controls that shape what a deployed AI system can access and do---access
controls, data segmentation, privilege boundaries, and containment design that
limit how far a breach can propagate. These controls, rather than cybersecurity
spending alone, shape how safely capable AI can be deployed and how costly the
gap becomes when governance lags.

Indeed, this gap carries a direct economic cost. The average cost of a data
breach reached \$4.88 million in 2024 \citep{ibm2024}. In sectors with high
breach costs, such as healthcare and financial services, the AI systems that
promise substantial productivity gains are often deployed with broad authority
and deep data access---environments in which the gap can leave firms especially
exposed.

The IS literature examines how firms deploy AI as part of digital
transformation and shows that realizing AI's value requires organizational
governance that differs from traditional IT \citep{bharadwaj2013, vial2019,
collins2021ai, holmstrom2022, linmaruping2025}. Research on organizational
cybersecurity demonstrates that IT investment, governance arrangements, and
organizational complexity shape breach outcomes \citep{li2023it,
liu2020centralized, durcikova2024united, tanriverdi2025taming}. What remains
less developed across these streams is a clear model of how the
governance-capability gap reshapes firms' deployment incentives and affects
whether capability improvements help or hurt deployment.

We address this question by jointly modeling AI deployment and cybersecurity
investment under the governance-capability gap. The key premise is not that AI
capability mechanically causes greater breach damage. Rather, some high-value
uses of more capable AI require greater authority exposure. When governance
controls have not matured enough to separate capability from that exposure,
the authority attached to capable systems increases the damage that can follow
from compromise. We refer to this organizational condition as
\textit{capability--damage bundling}. Throughout the paper, this term refers
to capability-linked authority exposure under governance lag, not to a direct
technological effect of capability on damage.

This bundling defines the model's scope. We do not claim that capability and
authority exposure always move together. The IBM and CSA surveys are
consistent with this condition, but they do not directly identify
capability--authority co-variation \citep{ibm2025, csa2026}. We therefore
study a specific setting: firms deploy more capable AI with greater authority
exposure before governance controls have matured. In this setting, capability
improvements can increase the damage that follows from compromise. Firms can
weaken or eliminate the condition through data segmentation, least-privilege
architecture, access-control maturation, and containment design.

Our central result shows a deployment paradox. In high-loss environments, an
improvement in AI capability can reduce, rather than increase, optimal
deployment when that capability is deployed with greater authority exposure.
More capable AI raises productivity \citep{brynjolfsson2023generative,
acemoglu2024simple}. Yet under capability--damage bundling it also raises the
security burden attached to deployment. When the marginal increase in the
expected security burden exceeds the marginal productivity gain, the firm
rationally deploys less. The paradox is therefore not that cyber risk lowers
deployment. It is that, after the firm optimally adjusts security investment,
an improvement in AI capability can still reduce optimal deployment.

The mechanism has two components. Broader deployment expands the attack
surface and increases the defensive expenditure needed to sustain a given
security posture. Greater capability increases the damage conditional on
compromise only when capability is bundled with greater authority exposure.
In sufficiently high-loss environments, these security effects can dominate
the productivity effect over an economically relevant range of capability
levels. The paradox then resolves either when capability becomes high enough
for productivity to dominate again, or when governance investment reduces the
loss environment that the firm faces.

The model generates three further implications. Optimal deployment lies below
the productivity-maximizing benchmark by a security discount that grows with
both breach-loss magnitude and capability. Governance investments that reduce
breach-loss magnitude shrink the paradox region itself. Breach externalities
expand the social paradox region by increasing the loss environment relevant
for welfare. The result also identifies its own boundary condition: when
governance investment decouples capability from authority exposure, or when
the loss environment is sufficiently mild, the standard monotone relationship
between capability and deployment is restored.

We contribute to IS research in three ways. First, to the literature on AI
deployment and digital transformation, we show a new comparative-static
possibility under organizational risk \citep{bharadwaj2013, vial2019,
collins2021ai, holmstrom2022, linmaruping2025}. Existing models imply that
better AI weakly increases deployment \citep{acemoglu2022artificial,
agrawal2018prediction, agrawal2019exploring}. Our model shows that this
monotonicity breaks down when governance maturity lags capability advancement:
capability improvements can also raise the security burden of deployment when
they are implemented with greater authority exposure under the governance gap.

Second, we endogenize the firm's cyber risk environment through its technology
choice. This connects IS research on organizational cybersecurity
\citep{li2023it, liu2020centralized, durcikova2024united, tanriverdi2025taming}
to AI deployment decisions. In the canonical framework of
\citet{gordon2002economics}, breach risk is exogenous. Here, AI deployment
under the gap determines both attack surface and breach consequences, creating
a feedback loop that prior models do not capture.

Third, we characterize the deployment implications of the governance-capability
gap as a structural organizational condition. We show that once the gap creates
capability--damage bundling, capability improvements can increase the
deployment-related security burden faster than productivity gains, generating
a deployment paradox. The core comparative static does not require a unique
micro-foundation for the gap. It requires a weaker condition: under governance
lag, more capable deployments carry greater authority exposure. By isolating
this condition, the model clarifies when stronger AI expands deployment and
when it can instead make deployment less attractive.

\section{Related Literature}

This work builds on three streams of research.

\paragraph{AI Deployment and Digital Transformation.}
The IS literature has examined digital business strategy, digital
transformation, and the organizational conditions that shape AI deployment
\citep{bharadwaj2013, vial2019, collins2021ai, holmstrom2022}. Realizing
AI's value requires more than technical adoption; it demands complementary
changes in governance structures, workflow redesign, and capability
development \citep{linmaruping2025}. Our model builds on this work by showing
how the governance-capability gap can reverse the usual relationship between
capability and deployment. When governance lags, more capable systems can
become more costly to integrate deeply because they require stronger
containment and access-control safeguards.

Analytical models of AI deployment have generally predicted that stronger
capability increases deployment value \citep{acemoglu2022artificial,
agrawal2018prediction, agrawal2019exploring}. Our model departs from this
prediction. Under the gap, many high-value AI deployments create value by
being deeply embedded in workflows and granted broader decision rights---yet
these same features can also amplify breach damage. The model identifies
conditions under which this tension produces a sign reversal.

\paragraph{Organizational Cybersecurity.}
IS research shows that cybersecurity outcomes depend on IT investment,
governance arrangements, organizational empowerment, and structural complexity
\citep{li2023it, liu2020centralized, durcikova2024united, tanriverdi2025taming}.
For example, \citet{tanriverdi2025taming} find that enterprise-wide data
analytics platforms can reduce complexity and improve security outcomes, while
\citet{liu2020centralized} show that centralized IT decision-making lowers
breach likelihood. These studies highlight how organizational design shapes
the security consequences of technology deployment. Related analytical work in
information security has also examined how firms evaluate security investment
under economic tradeoffs, using both decision-theoretic and game-theoretic
approaches to characterize optimal protection choices
\citep{gordon2002economics, cavusoglu2004model, cavusoglu2008decision}.
More broadly, the economics of privacy and information security emphasizes that
data exposure can create substantial organizational and social costs
\citep{acquisti2016economics}.

We take a different approach. Under gap conditions, the firm's deployment
decision itself jointly determines both the attack surface and the
consequences of a breach. Unlike \citet{gordon2002economics}, where
vulnerability is treated as exogenous, our formulation makes both the
attack-surface term and the breach-loss magnitude term endogenous to the
firm's technology choice under the governance-capability gap. The central
contrast is not only that risk becomes endogenous. A capability improvement
can reverse the deployment comparative static when technology choice jointly
changes exposure and conditional loss.

\paragraph{Dual-Use AI and Governance.}
AI simultaneously expands capability and introduces new organizational risk
\citep{choi2025dual, hendrycks2023overview}. Related IS work emphasizes that
managers must actively govern AI deployment \citep{berente2021managingai}.
Our model offers a tractable explanation of why governance can become a binding
constraint: when maturity lags capability, the same improvement that raises
productivity can also increase the security burden when it is implemented with
greater authority exposure. In high-loss environments, this dynamic can make
frontier AI deployment value-reducing and identifies governance investment as
one mechanism for realizing AI's productive value.

Taken together, these streams explain how firms obtain value from AI, how
organizational design shapes cybersecurity outcomes, and how firms choose
security investment under economic tradeoffs. They say less about settings in
which the technology choice itself changes the firm's cyber risk environment.
Our model focuses on this case. Under the governance-capability gap, AI
deployment determines both the attack surface and the conditional loss
environment through authority exposure. This authority-exposure linkage
generates a comparative static that prior frameworks do not typically predict:
stronger capability can reduce, rather than increase, optimal deployment.

Methodologically, the paper follows the tradition of parsimonious analytical
models that derive counterintuitive comparative statics from minimal structure,
as in \citet{bakos1999bundling}, \citet{choudhary2007software}, and
\citet{bova2022quantum}.

\section{Model}

This section formalizes the firm's joint deployment and security problem. The
firm chooses AI deployment $\alpha$ and security investment $d$. The
environment is characterized by AI capability $\theta$, conditional
breach-loss magnitude $\lambda$, and complementary organizational readiness
$\mu$. The governance-capability gap enters through an authority-exposure
index $a(\theta)$, which links capability to conditional breach damage when
governance has not decoupled capability from operational exposure.

\subsection{Setup}
\label{subsec:setup}

Consider a firm deciding how deeply to deploy AI and how much to invest in
cybersecurity. Not all AI deployment carries the same security implications.
Narrow, siloed tools---operating on restricted data without cross-system
authority---add little to cyber exposure. The key tension is the
\textit{governance-capability gap}: high-value uses of more capable AI are
often implemented with broader authority exposure than governance controls
can safely contain, because access-control and containment maturity has not
kept pace. Under this condition, capability can be associated with greater
authority exposure in practice, and we refer to this linkage as
\textit{capability--damage bundling}.

\paragraph{Two margins of the firm's security problem.}
We distinguish two margins. \textit{Security investment} $d \geq 0$ is a
within-period flow: monitoring, threat detection, incident response, and
access-control enforcement that reduces breach probability given the current
deployment level. \textit{Governance maturity} determines $\lambda$: the
inherited architecture of access controls, data segmentation, privilege
boundaries, and containment design that governs how much damage a successful
breach can cause. The governance-capability gap refers to the slow adjustment
of this second margin relative to the pace of capability advancement.

\paragraph{Conceptual map of the model's parameters.}
The model has four key parameters. $\alpha \geq 0$ captures \textit{deployment
scope}---the organizational breadth over which AI is integrated. $\theta > 0$
indexes \textit{AI capability level}---the productive power of each deployed
unit. The governance gap links this capability to operational exposure, denoted
by an authority-exposure index $a(\theta)$, which summarizes the operational
exposure attached to a deployed system of capability $\theta$. $\lambda > 0$
captures \textit{conditional breach loss magnitude}---residual loss severity
given compromise, reflecting containment quality, data segmentation, and
regulatory liability. In the baseline model, we treat $\lambda$ as exogenous
to isolate how AI capability affects deployment under a given security
environment; Section~\ref{sec:productivity} and
Appendix~\ref{app:endogenous_gov} discuss this interpretation in more detail.

$\mu > 0$ captures \textit{complementary organizational readiness}---productivity
gains from process standardization, workflow redesign, and organizational
learning that raise deployment value without expanding the damage channel.

The $\mu$--$\lambda$ distinction matters for strategy. $\mu$-raising
investments work on the productivity side; $\lambda$-reducing investments work
on the security side. Both improve performance, but through different channels.

\paragraph{Formalizing capability--damage bundling.}
The model's central organizational condition is that high-value uses of more
capable AI may require greater operational exposure when governance lags
capability. We formalize this condition as follows.

\begin{assumption}[Capability-Linked Authority Exposure under Governance Lag]
\label{assump:bundling}
Under the governance-capability gap, high-value deployment of more capable AI
systems requires greater authority exposure. Let $a(\theta)$ denote the
authority-exposure index associated with capability $\theta$, with
$a'(\theta)\ge 0$. The baseline model sets $a(\theta)=\theta$, so conditional
breach damage takes the tractable form $L(\alpha,\theta)=\lambda\alpha\theta$.
\end{assumption}

\noindent Assumption~\ref{assump:bundling} is a scope condition, not an
empirical regularity claimed to hold for all AI systems. The monotonicity of
$a(\theta)$ defines the class of deployments studied here: settings in which
firms pursue the productive value of more capable AI by embedding the system
more deeply into organizational workflows. When governance decouples capability
from authority exposure, the relevant case approaches $a'(\theta)=0$, and the
deployment paradox weakens or disappears. This boundary condition also gives
the mechanism empirical content: the model predicts no deployment paradox in
settings where capability improvements do not carry greater authority exposure.
The condition $a'(\theta)>0$ alone does not imply the deployment paradox. The
sign reversal arises only when the marginal increase in the exposure-linked
security burden exceeds the marginal productivity gain over some capability
range.

A useful way to interpret the baseline loss function is through this
authority-exposure index. Conditional breach damage can be written as
$L(\alpha,\theta)=\lambda\alpha a(\theta)$. The baseline normalization
$a(\theta)=\theta$ does not mean that technical capability itself mechanically
causes damage. It means that, in the organizational environments studied here,
higher-capability deployments are implemented with greater operational exposure
unless governance controls decouple the two. Appendix~\ref{app:generalized_L}
sets $a(\theta)=\theta^\gamma$ and shows how the paradox changes as the
elasticity of authority exposure with respect to capability varies.

\paragraph{The firm's problem.}
The firm chooses two instruments: AI deployment level ($\alpha \geq 0$;
deployment scope) and security investment ($d \geq 0$; defensive expenditure).
The firm's profit is:
\begin{equation}
    \pi(\alpha, d) =
        \underbrace{(\theta + \mu)\alpha - \frac{\alpha^2}{2}}_{\text{AI productivity}}
        - \underbrace{\frac{\alpha}{\alpha + d} \cdot \lambda\alpha\theta}_{
        \text{Expected security loss}}
        - \underbrace{d}_{\text{Security investment}}
    \label{eq:profit}
\end{equation}

Deployment $\alpha$ enters through two channels: it raises both breach
probability (via the attack-surface term $\alpha/(\alpha+d)$) and breach
damage (via the propagation-scale term $\lambda\alpha\theta$).\footnote{The
model focuses on incremental cyber risk from AI deployment. Baseline non-AI
risk is treated as exogenous and separable, standard in security economics,
to isolate how the gap shapes the firm's marginal security environment. By
convention, $p(0,d)=0$ for all $d\ge 0$, so that zero deployment yields zero
incremental breach probability.} Profit should be interpreted as expected
present value over the focal decision horizon: current deployment benefits,
current defensive expenditure, and the discounted expected loss from future
breach consequences are all summarized in the single-period objective.

The joint optimization of $(\alpha, d)$ reflects the standard treatment in
security economics. \citet{gordon2002economics}---the canonical framework our
model extends---jointly optimizes security investment and breach probability as
a single integrated decision; we follow that convention and additionally allow
the technology choice $\alpha$ to endogenize the cyber risk environment. The
simultaneity also reflects observed enterprise practice: security architecture
reviews, access-control scoping, and deployment planning are integrated into
AI adoption rather than treated sequentially \citep{ibm2025, csa2026}. The
ex ante framing captures this integration---the firm anticipates breach
probability and damage conditional on its chosen $\alpha$ before committing
to it.

\subsection{AI Productivity}
\label{sec:productivity}

The term $(\theta + \mu)\alpha - \frac{\alpha^2}{2}$ captures the net
productivity benefit of AI deployment. It has three components.

\paragraph{Capability-driven value ($\theta\alpha$).}
$\theta$ captures the capability level of each deployed AI unit. Higher
$\theta$ generates greater value through prediction, automation, and decision
support \citep{brynjolfsson2023generative, acemoglu2024simple}. The model
separates this productivity effect from the authority-exposure channel in
Assumption~\ref{assump:bundling}: capability raises productivity directly, and
it raises conditional breach damage only when deployed with broader authority
exposure under governance lag.

\paragraph{Complementary organizational readiness ($\mu\alpha$).}
$\mu$ captures the productivity gains from process standardization, workflow
integration, and organizational learning \citep{lee2022ai}. These investments
improve deployment value without expanding the damage channel, because they do
not co-vary with authority exposure under the gap. This is why $\mu$ enters
productivity but not the loss function. A richer specification could allow
limited spillover from $\mu$ into the damage channel; this changes the location
of the paradox region but does not overturn the core sign-reversal result
(Appendix~\ref{app:mu_damage}).

\paragraph{Convex deployment cost ($\alpha^2/2$).}
The quadratic term captures the convex organizational costs of workflow
redesign, integration friction, and coordination overhead; we normalize the
scaling parameter to 1 without loss of generality, as rescaling is absorbed
into the units of $\alpha$. Without cybersecurity concerns, the firm would
choose $\alpha^0 = \theta + \mu$. The deployment variable $\alpha$ is
interpreted as a deployment-intensity index rather than a literal share bounded
between 0 and 1; the quadratic term ensures that gross AI productivity is
eventually offset by coordination and integration frictions.

\paragraph{The exogeneity of $\lambda$.}
We treat $\lambda$ as exogenous because it represents the firm's inherited
breach-loss environment rather than a contemporaneous choice. Security
investment $d$ is a within-period flow: the firm adjusts it in response to
current deployment and exposure conditions, reducing breach probability through
monitoring, detection, and incident response. $\lambda$ moves more slowly. It
reflects accumulated segmentation architecture, privilege design, containment
quality, and propagation boundaries---features built up over years of prior
investment that do not adjust one-for-one with current security spending.
Governance restructuring faces sharply rising costs: tightening containment
architecture becomes disproportionately expensive at the margin. Even a
forward-looking firm cannot close the gap overnight. We therefore study the
within-period problem in which the firm chooses $\alpha$ and $d$ taking
$\lambda$ as given. Accordingly, $\lambda$ should be interpreted as the firm's
ex ante expected conditional breach-loss environment at the time of the
deployment decision, not as a realized breach outcome known with certainty.

This abstraction isolates a within-period question: given the inherited loss
environment a firm currently faces, how do deployment and security investment
respond? Appendix~\ref{app:endogenous_gov} shows that when governance
restructuring entails convex organizational costs, a firm with a positive
incentive to reduce $\lambda$ may rationally leave a residual gap in
equilibrium---what we term \textit{rational security debt}. Comparative
statics with respect to $\lambda$ should therefore be read as comparisons
across inherited governance environments, or as the long-run effect of prior
governance investment that shifts the loss environment before the deployment
decision. In either interpretation, closing the gap shrinks the paradox region
and raises optimal deployment.%
\footnote{$\lambda$ should be read as a reduced-form index with both
industry-level and firm-level components. The calibration in
Section~\ref{sec:numerical} uses industry-level breach-cost data, capturing
the former. The discussion in Section~\ref{sec:discussion} emphasizes the
latter: over longer horizons, firms can reduce $\lambda$ through investments
in containment quality, data segmentation, and post-breach propagation
control.}

\subsection{Attack Surface and Breach Probability}

\begin{equation}
    p(\alpha, d) = \frac{\alpha}{\alpha + d}
    \label{eq:breach_prob}
\end{equation}
gives the probability of successful compromise over the AI-enabled incremental
exposure margin, holding baseline non-AI security fixed, and is modeled as a
standard Tullock contest function \citep{gordon2002economics, tullock1980}.
Under the baseline Tullock form, it can also be interpreted as the fraction
of AI-attributable exposure left undefended. As deployment deepens under the
gap, each additional AI agent, API connection, and autonomous workflow creates
a new potential entry point, so the attack surface expands with
$\alpha$.%
\footnote{We adopt the baseline Tullock form ($\beta=1$) for tractability.
Section~\ref{subsec:beta} and Appendix~\ref{app:generalized_p} show that under
$\beta>1$---when governance failures allow superlinear attack-surface
growth---the attack-surface channel can additionally strengthen sign reversal
locally under interior feasibility.} In the contest formulation, $\alpha$ and
$d$ enter as dimensionless, normalized intensities of exposure and defense
rather than raw physical or monetary units. Their ratio therefore captures the
relative balance between the attack surface created by deployment and the
defensive coverage allocated against it.

\subsection{Breach Damage}

Let $a(\theta)$ denote the authority-exposure index attached to a deployed
system of capability $\theta$ under governance lag. Conditional breach damage
can be written as $L(\alpha,\theta)=\lambda\alpha a(\theta)$. The baseline
model sets $a(\theta)=\theta$, yielding
\begin{equation}
    L(\alpha, \theta) = \lambda\alpha a(\theta)=\lambda\alpha\theta.
    \label{eq:loss}
\end{equation}
This expression represents reduced-form incremental breach damage under the
gap. $\lambda$ is the \textit{conditional breach loss magnitude}---reflecting
containment quality, data segmentation, and regulatory liability. Broader
deployment ($\alpha$) increases propagation scale. Capability affects damage
only through the authority exposure attached to deployment under
Assumption~\ref{assump:bundling}. By contrast, investments associated with
$\mu$ raise productivity without expanding the damage channel.\footnote{$\lambda$
governs only the loss channel $L$; the breach-probability channel
$p(\alpha,d)$ is modeled separately.}

\subsection{Expected Security Loss}

\begin{equation}
    \text{Expected Security Loss} = p(\alpha, d) \cdot L(\alpha, \theta) =
    \frac{\lambda\alpha^2\theta}{\alpha + d}
    \label{eq:expected_loss}
\end{equation}

The expected-loss formulation follows standard practice in security economics.
A risk-averse firm would additionally penalize variance in breach outcomes,
which would widen the paradox region further. The expected-value approach
therefore provides a lower bound on the security cost facing organizations
with standard risk preferences.

The paper focuses on the deployment problem in environments where positive AI
deployment remains economically relevant. The condition $\lambda < \mu + 1$
formalizes this boundary. When it fails, breach-loss magnitude is so severe
relative to complementary organizational readiness that the firm optimally
shuts down AI deployment altogether ($\alpha^*=0$). That is a straightforward
corner solution, not the sign-reversal paradox that is the paper's focus.

An asterisk denotes an optimal value throughout: $\alpha^*$, $d^*$, and $p^*$
refer to the firm's optimal deployment, defensive expenditure, and equilibrium
breach probability, respectively.

\begin{assumption}[Positive Deployment]
    \label{assump:interior}
    $\lambda < \mu + 1$, ensuring $\alpha^* > 0$ for all $\theta > 0$.
\end{assumption}

The condition defines the economic boundary for positive deployment: a firm
with higher organizational readiness $\mu$ can sustain positive deployment
under a wider range of loss environments, while a firm facing very high
breach-loss magnitude $\lambda$ may find even modest AI deployment
unprofitable. Assumption~\ref{assump:interior} excludes only the extreme case
in which the gap is so severe that no deployment is worthwhile at any
capability level.

When $\lambda>1$, this condition ensures the global minimum of $\alpha^*$,
attained at $\theta=\lambda$, remains strictly positive:
$\alpha^*_{\min} = \mu + 1 - \lambda > 0$ (Appendix~\ref{app:global}).
When $\lambda\le 1$, $\alpha^*$ is increasing in $\theta$ throughout and the
positive-deployment condition is automatically satisfied.

When $\lambda\theta > 1$, the firm invests in security ($d^* > 0$); when
$\lambda\theta \le 1$, $d^* = 0$. The boundary $\theta = 1/\lambda$ separates
these regimes.

\section{Analysis}

\subsection{Optimal Security Investment}

\begin{lemma}[Optimal Security Investment]
    \label{lem:dstar}
    For a given deployment level $\alpha > 0$:
    \begin{equation}
        d^*(\alpha, \theta) =
        \begin{cases}
            0, & \lambda\theta \le 1, \\[4pt]
            \alpha\bigl(\sqrt{\lambda\theta} - 1\bigr), & \lambda\theta > 1.
        \end{cases}
        \label{eq:d_star}
    \end{equation}
\end{lemma}

\begin{proof}
    Evaluating $\partial\pi/\partial d$ at $d=0$ yields $\lambda\theta - 1$.
    If $\lambda\theta \le 1$, the marginal benefit of defense never exceeds
    its cost, so the non-negativity constraint binds and $d^*=0$. Economically, 
    this is the case in which the combined effect of capability and breach-loss 
    magnitude is too small to justify dedicated defensive expenditure.

    If $\lambda\theta > 1$, setting the first-order condition $\partial\pi/\partial d = 0$ yields
    \[
    \frac{\lambda \alpha^2 \theta}{(\alpha+d)^2}=1
    \quad\Longrightarrow\quad
    \alpha+d=\alpha\sqrt{\lambda\theta}
    \quad\Longrightarrow\quad
    d = \alpha(\sqrt{\lambda\theta}-1) > 0.
    \]
    Because $\partial^2\pi/\partial d^2 = -2\lambda\alpha^2\theta/(\alpha+d)^3 < 0$, 
    strict concavity in $d$ confirms this is a unique global maximum. See 
    Appendix~\ref{app:proof_dstar} for the complete derivation.
\end{proof}

The formula $d^* = \alpha(\sqrt{\lambda\theta}-1)$ has a direct operational
reading. For each unit of AI deployed, the firm must commit
$\sqrt{\lambda\theta}-1$ units of defensive expenditure at the interior
optimum. This multiplier rises with both breach-loss magnitude $\lambda$ and
capability $\theta$: under gap conditions, more capable systems require more
defense per unit of deployment to hold breach probability constant. A firm
upgrading from $\theta_L$ to $\theta_F > \theta_L$ therefore faces a higher
per-unit defense burden even before the productivity gain materializes. When
$\lambda\theta \le 1$, the expected marginal benefit of AI-specific defense
falls below its cost, and the firm allocates zero to that margin.

\subsection{Breach Probability Under Optimal Defense}

\begin{lemma}[Equilibrium Breach Probability]
    \label{lem:pstar}
    Under optimal security investment,
    \begin{equation}
        p^*(\theta) =
        \begin{cases}
            1, & \lambda\theta \le 1, \\[4pt]
            \dfrac{1}{\sqrt{\lambda\theta}}, & \lambda\theta > 1.
        \end{cases}
        \label{eq:p_star}
    \end{equation}
    In neither regime does $p^*$ depend on $\alpha$.
\end{lemma}

\begin{proof}
    Corner: If $\lambda\theta\le 1$, then $d^*=0$, so $p^*=\alpha/\alpha=1$. 
    At this corner, the AI-attributable incremental exposure margin is left fully undefended.
    
    Interior: If $\lambda\theta>1$, substituting $d^*=\alpha(\sqrt{\lambda\theta}-1)$ into
    $p(\alpha,d)=\alpha/(\alpha+d)$ gives:
    \[
    p^* = \frac{\alpha}{\alpha+\alpha(\sqrt{\lambda\theta}-1)} = \frac{\alpha}{\alpha\sqrt{\lambda\theta}} = \frac{1}{\sqrt{\lambda\theta}}.
    \]
    In both regimes $\alpha$ cancels, so $p^*$ is independent of deployment. See 
    Appendix~\ref{app:proof_pstar} for the complete derivation.
\end{proof}

When $\lambda\theta \le 1$, the corner solution $p^*=1$ requires careful
interpretation. It does \emph{not} mean that the firm suffers a firm-wide
breach with certainty. Rather, $p^*$ measures exposure over the AI-attributable
incremental margin. At the corner, the firm assigns no defensive effort to
that margin while baseline non-AI security remains fixed. Thus, $p^*=1$ means
full exposure of the incremental AI-specific margin, not certainty of a
firm-wide breach. The corner solution reflects an economic decision: when the
marginal cost of defense exceeds the expected benefit from reducing breach
probability over the AI-enabled margin, no investment in AI-specific defense
is worthwhile. At the boundary $\theta=1/\lambda$, both expressions yield
$p^*=1$.\footnote{$\partial p^*/\partial\theta=0$ in the corner and
$-\frac{1}{2}\sqrt{\lambda/\theta^3}$ in the interior; this does not affect
the smoothness of $\alpha^*$ or $\Delta$, which are $C^1$ at $\theta=1/\lambda$
(Proposition~\ref{prop:discount}).}

The independence of $p^*$ from $\alpha$ under $\beta=1$ reflects an endogenous
organizational response. As deployment expands the attack surface, a rational
firm scales its defensive expenditure $d^*$ proportionally, so the probability
effect cancels in equilibrium. Deployment therefore does not create an
additional comparative-static effect through equilibrium breach probability.

What deeper deployment changes is the security burden required to sustain that
probability. Specifically, the firm must bear both the defense bill
$d^*=\alpha(\sqrt{\lambda\theta}-1)$ and the residual damage
$p^*\!\cdot\!L=\alpha\sqrt{\lambda\theta}$. The sign reversal in
Proposition~\ref{prop:paradox} therefore arises from the
authority-exposure-amplified security burden, not from an increase in
equilibrium breach probability itself. Section~\ref{subsec:beta} shows that
when $\beta>1$, the attack-surface channel can additionally reinforce this
effect.

\subsection{Reduced-Form Profit and Optimal Deployment}

\begin{lemma}[Reduced-Form Profit]
    \label{lem:reduced}
    \begin{equation}
        \pi(\alpha;\theta) =
        \begin{cases}
            \bigl[\mu + \theta(1-\lambda)\bigr]\alpha - \dfrac{\alpha^2}{2},
            & \lambda\theta \le 1, \\[8pt]
            \bigl[\theta + \mu + 1 - 2\sqrt{\lambda\theta}\bigr]\alpha
                - \dfrac{\alpha^2}{2},
            & \lambda\theta > 1.
        \end{cases}
        \label{eq:reduced_profit}
    \end{equation}
\end{lemma}

\begin{proof}
    In each regime, substitute the optimal $d^*$ and $p^*$ from
    Lemmas~\ref{lem:dstar}--\ref{lem:pstar} into $\pi(\alpha,d)$ and collect
    terms in $\alpha$. See Appendix~\ref{app:proof_lemma3} for the complete
    derivation.
\end{proof}

\begin{proposition}[Optimal AI Deployment]
    \label{prop:alpha_star}
    Under Assumption~\ref{assump:interior}:
    \begin{equation}
        \alpha^*(\theta) =
        \begin{cases}
            \mu + \theta(1-\lambda), & \lambda\theta \le 1, \\[4pt]
            \theta + \mu + 1 - 2\sqrt{\lambda\theta}, & \lambda\theta > 1,
        \end{cases}
        \label{eq:alpha_star}
    \end{equation}
    with $\alpha^*(\theta)>0$ for all $\theta>0$. Both expressions coincide
    at $\theta=1/\lambda$.
\end{proposition}

\begin{proof}
    Each regime's reduced-form profit from Lemma~\ref{lem:reduced} is a concave quadratic in $\alpha$
    ($\partial^2\pi/\partial\alpha^2=-1<0$), so the FOC is sufficient. The two expressions coincide at
    $\theta=1/\lambda$. When $\lambda>1$, the global minimum occurs at
    $\theta=\lambda$ (Appendix~\ref{app:global}). See
    Appendix~\ref{app:proof_alpha_star} for the complete derivation.
\end{proof}

Lemma~\ref{lem:reduced} shows that once the firm chooses defensive spending
optimally within a given regime, its remaining decision is how broadly to
deploy AI ($\alpha$). In each regime, the firm therefore chooses $\alpha$ by
weighing the marginal productivity gain from broader deployment against the
marginal security burden implied by that regime. In the corner regime, where
AI-specific defense is not yet worthwhile, deployment rises with complementary
organizational readiness $\mu$ and rises with capability only when the loss
environment is sufficiently mild. In the interior regime, the firm continues
to broaden deployment only if the added productivity from greater capability
exceeds the combined burden of residual breach damage and defensive
expenditure. The two regime-specific expressions connect smoothly at the
boundary because the firm changes its response rule once defense becomes
worthwhile, not because deployment jumps discontinuously.

\subsection{Main Results}

\begin{proposition}[Security Discount]
    \label{prop:discount}
    Let $\alpha^0=\theta+\mu$. The security discount
    $\Delta\equiv\alpha^0-\alpha^*$ is
    \begin{equation}
        \Delta(\theta) =
        \begin{cases}
            \lambda\theta, & \lambda\theta \le 1, \\[4pt]
            2\sqrt{\lambda\theta}-1, & \lambda\theta > 1.
        \end{cases}
        \label{eq:discount}
    \end{equation}
    $\Delta$ is continuously differentiable at $\theta=1/\lambda$, increasing
    in $\lambda$ and $\theta$, and independent of $\mu$.
\end{proposition}

\begin{proof}
    Direct computation of $\Delta\equiv\alpha^0-\alpha^*$ using
    $\alpha^0=\theta+\mu$ and Proposition~\ref{prop:alpha_star}. Both
    expressions yield $1$ at $\theta=1/\lambda$ and their derivatives match
    at that boundary. See Appendix~\ref{app:proof_discount} for the complete
    derivation.
\end{proof}

From the firm's perspective, the security discount is the wedge between what
it would deploy if cyber risk were irrelevant and what it can justify once the
governance-capability gap is taken into account. As capability advances, the
firm's no-risk deployment benchmark rises, but so does the amount of deployment
withheld because the security burden also grows. The discount is therefore not
a fixed cost of adopting AI. It is an endogenous gap between the firm's no-risk
deployment benchmark and its economically sustainable deployment under the
current loss environment. Figure~\ref{fig:discount_app} in Appendix~\ref{app:comp_statics} visualizes
this security discount, and Table~\ref{tab:comp_statics} summarizes the
interior comparative statics for deployment, defense, breach probability, and
firm value.

\begin{remark}
    The security discount widens as capability advances. Stronger AI therefore
    raises the firm's no-risk deployment benchmark and, at the same time,
    increases the amount of deployment withheld under the gap. Investments that
    reduce $\lambda$ become more valuable as capability advances because they
    narrow this wedge.
\end{remark}

\begin{proposition}[Deployment Paradox]
    \label{prop:paradox}
    Suppose $\lambda<\mu+1$. Under capability--damage bundling,
    \begin{equation}
        \frac{\partial\alpha^*}{\partial\theta} =
        \begin{cases}
            1-\lambda, & \lambda\theta \le 1, \\[4pt]
            1-\sqrt{\lambda/\theta}, & \lambda\theta > 1.
        \end{cases}
        \label{eq:paradox}
    \end{equation}
    If $\lambda>1$: (i) optimal deployment is strictly decreasing in
    capability for all $\theta\in(0,\lambda)$; (ii) strictly increasing for
    $\theta>\lambda$; (iii) the global minimum $\alpha^*_{\min}=\mu+1-\lambda$
    is attained at $\theta=\lambda$. If $\lambda\le 1$:
    $\partial\alpha^*/\partial\theta\ge 0$ for all $\theta$, and no paradox
    arises.
\end{proposition}

\begin{proof}
    Differentiate $\alpha^*(\theta)$ from Proposition~\ref{prop:alpha_star}
    with respect to $\theta$ in each regime. The sign of
    $\partial\alpha^*/\partial\theta$ determines the paradox region. Since
    $1/\lambda<\lambda$ when $\lambda>1$, the corner and interior paradox
    intervals adjoin, covering $(0,\lambda)$. See
    Appendix~\ref{app:proof_paradox} for the complete derivation.
\end{proof}

This sign reversal is not an artifact of the square-root form in the baseline
interior solution. Appendix~\ref{app:generalized_L} and
Appendix~\ref{app:generalized_prod} show that related sign-reversal regions
can arise under generalized damage and productivity functions whenever the
exposure-linked security burden grows sufficiently strongly relative to the
marginal productivity gain.

\begin{remark}
    The deployment paradox is a rational organizational response to the gap,
    not an irrationality or misperception. In the corner regime, each
    capability increment adds $1-\lambda<0$ to deployment when $\lambda>1$.
    In the interior regime, gross productivity grows linearly with $\theta$
    while the expected security burden scales with $\sqrt{\lambda\theta}$;
    the turning point at $\theta=\lambda$ marks where the linear productivity
    channel begins to dominate. The paradox resolves either when capability
    crosses $\theta=\lambda$, or when $\lambda$-reducing governance investment
    lowers $\lambda$ below 1 (Proposition~\ref{prop:complementarity}). The
    boundary $\theta=1/\lambda$ marks where defensive investment becomes
    worthwhile, not where the paradox emerges or resolves.
\end{remark}

The mechanism also has a straightforward cross-industry interpretation. Firms
facing the same capability improvement can respond differently depending on
their loss environment. In a high-loss environment, a more capable system is
deployed with broader authority exposure under the gap, so the damage
amplification from capability growth can outpace the productivity gain over an
intermediate range. The firm then responds by reducing its deployment footprint.
In a milder loss environment, by contrast, the same capability improvement
raises deployment because the productivity channel continues to dominate. The
paradox is therefore not that the AI system is undesirable in absolute terms;
it is that the firm's current controls may be unable to safely contain what the
more capable system can reach.

\begin{proposition}[Effect of Governance Maturity on Deployment]
    \label{prop:complementarity}
    A reduction in $\lambda$, reflecting governance investment that closes the
    gap, increases optimal deployment:
    \begin{equation}
        \frac{\partial\alpha^*}{\partial\lambda} =
        \begin{cases}
            -\theta, & \lambda\theta \le 1, \\[4pt]
            -\sqrt{\theta/\lambda}, & \lambda\theta > 1.
        \end{cases}
    \end{equation}
    Both are strictly negative; both yield $-1/\lambda$ at $\theta=1/\lambda$.
\end{proposition}

\begin{proof}
    Direct differentiation of \eqref{eq:alpha_star} with respect to $\lambda$.
    See Appendix~\ref{app:proof_complementarity} for the complete derivation.
\end{proof}

This comparative static has a direct managerial interpretation. A reduction in
$\lambda$ does not simply lower expected breach losses holding deployment fixed.
It also relaxes the security burden that constrains how broadly the firm can
use capable AI in the first place. Governance investments that improve
containment, segmentation, and privilege design can therefore operate as
enabling investments rather than as separate compliance costs. By reducing the
loss environment the firm faces, they convert more of AI's technical capability
into economically viable deployment. Figure~\ref{fig:complement} in Appendix~\ref{app:comp_statics} illustrates
this comparative static by plotting optimal deployment as the loss environment
changes.

\begin{remark}
    Reducing $\lambda$ below 1 eliminates the paradox region entirely:
    $\partial\alpha^*/\partial\theta\ge 0$ for all $\theta$ once $\lambda\le
    1$. Closing the gap therefore does more than increase deployment at the
    margin. It restores the standard monotone relationship between capability
    and deployment by moving the firm out of the paradox regime altogether.
\end{remark}

\subsection{Attack-Surface Channel Under Severe Governance Failure}
\label{subsec:beta}

The baseline $\beta=1$ specification is intentionally conservative. In this
case, the attack-surface channel cancels in equilibrium because the firm can
offset a broader attack surface by scaling defense proportionally with
deployment. The paradox therefore operates through the damage channel: more
capable AI raises conditional breach severity under the gap. This subsection
relaxes that proportional-offset property. Formally, consider
$p(\alpha,d)=\alpha^\beta/(\alpha^\beta+d)$ with $\beta>1$. Under this
specification, $p^*_\beta=\alpha^{(\beta-1)/2}/\sqrt{\lambda\theta}$ increases
in $\alpha$, so optimal defense can no longer fully offset superlinear
attack-surface growth. The attack-surface channel then begins to reinforce the
baseline damage-channel mechanism. A local sign-reversal condition via the
implicit function theorem is:
\[
    \frac{(\beta+1)\alpha^{(\beta-1)/2}\sqrt{\lambda}}{2\sqrt{\theta}}>1.
\]
This is a local result under interior feasibility, not a global theorem.
Appendix~\ref{app:generalized_p} provides the full derivation via the implicit
function theorem.

\paragraph{Asymmetric attacker advantage.}
The $\beta>1$ specification also has a natural interpretation in terms of
attacker--defender asymmetry. In practice, adversaries can probe many entry
points simultaneously while the defender must allocate resources across all of
them; a single successful intrusion suffices for the attacker, while the
defender must protect every node. This structural asymmetry means that the
effective attack-surface elasticity the defender faces may exceed one---each
additional AI-enabled entry point is harder to defend proportionally than the
last, because the attacker can concentrate effort while the defender cannot.
Under $\beta>1$, this asymmetry is partially captured: the condition
$(\beta+1)\alpha^{(\beta-1)/2}\sqrt{\lambda}/(2\sqrt{\theta})>1$ becomes
easier to satisfy as $\alpha$ grows, reflecting the compounding difficulty of
defending an expanding surface against a concentrated adversary.

The baseline $\beta=1$ should therefore be read as a conservative benchmark
for isolating the damage-channel mechanism; $\beta>1$ can additionally
strengthen sign reversal locally under interior feasibility.

\subsection{Endogenous Capability Choice: When Firms Rationally Reject
Frontier AI}
\label{subsec:endogenous_theta}

The main model treats $\theta$ as exogenous. However, the gap also shapes
capability selection: under gap conditions, deploying a more capable system
widens the damage channel, while retaining a legacy system limits exposure.
The same high-gap environments that generate the deployment paradox can also
distort capability choice itself. A firm may rationally reject a frontier
upgrade, and that rejection can persist even after the local paradox in
deployment has ended because firm value has not yet recovered.

Suppose the firm faces a discrete choice between legacy $\theta_L$ and frontier
$\theta_F>\theta_L$. Let $\theta_c \in \{\theta_L, \theta_F\}$ denote the
firm's choice of capability level. Substituting $\alpha^*(\theta_c)$ into the
reduced-form profit of Lemma~\ref{lem:reduced} yields
$\pi^*(\theta_c)=[\alpha^*(\theta_c)]^2/2$; we therefore define firm value as
$V(\theta_c)\equiv[\alpha^*(\theta_c)]^2/2$.

\begin{proposition}[Endogenous Capability Choice]
    \label{prop:endogenous}
    Under Assumption~\ref{assump:interior}:
    \begin{enumerate}
        \item \textbf{If $\lambda\le 1$:} $V(\theta_c)$ is weakly increasing
        for all $\theta_c>0$. The firm weakly prefers the higher-capability
        endpoint $\theta_F$.
        \item \textbf{If $\lambda>1$:} $V(\theta_c)$ is U-shaped with global
        minimum at $\theta_c=\lambda$. A firm choosing from
        $[\theta_L,\theta_F]$ always selects an endpoint:
        \begin{enumerate}
            \item \textbf{Both in the paradox region} ($\theta_F\le\lambda$):
            $V$ strictly decreasing; firm rejects the frontier upgrade.
            \item \textbf{Frontier beyond the turning point}
            ($\theta_F>\lambda$): the firm upgrades if and only if
            $V(\theta_F)>V(\theta_L)$. When both systems lie in the interior
            regime ($\lambda\theta_L>1$), this reduces to
            \[
                \sqrt{\theta_F}+\sqrt{\theta_L}>2\sqrt{\lambda},
            \]
            or, equivalently given $\theta_F>\theta_L$,
            \[
                \theta_F>
                \left[\max\{0,\,2\sqrt{\lambda}-\sqrt{\theta_L}\}\right]^2.
            \]
            When the legacy system lies in the corner regime
            ($\lambda\theta_L\le 1$), $V(\theta_L)=
            [\mu+\theta_L(1-\lambda)]^2/2$ and the upgrade condition
            $V(\theta_F)>V(\theta_L)$ must be evaluated directly.
        \end{enumerate}
    \end{enumerate}
\end{proposition}

\begin{proof}
    We derive $V(\theta_c) = [\alpha^*(\theta_c)]^2/2$ from
    Lemma~\ref{lem:reduced}: in each regime, profit is
    $C(\theta_c)\alpha - \alpha^2/2$ at the optimum $\alpha^* = C(\theta_c)$,
    giving $\pi^* = [C(\theta_c)]^2/2$. For $\lambda \le 1$,
    Proposition~\ref{prop:paradox} gives $\partial\alpha^*/\partial\theta \ge
    0$, so $V$ is weakly increasing. For $\lambda > 1$, computing
    $dV/d\theta_c$ in each regime shows that $V$ is strictly decreasing on
    $(0,\lambda)$ and strictly increasing on $(\lambda,\infty)$, with $C^1$
    continuity at $\theta_c=1/\lambda$ (Appendix~\ref{app:global}). The
    U-shape forces a boundary solution on $[\theta_L,\theta_F]$. When both
    systems lie in the interior regime, the upgrade condition reduces to
    $\sqrt{\theta_F}+\sqrt{\theta_L}>2\sqrt{\lambda}$, or equivalently, given
    $\theta_F>\theta_L$,
    $\theta_F>\left[\max\{0,\,2\sqrt{\lambda}-\sqrt{\theta_L}\}\right]^2$.
    See Appendix~\ref{app:proof_endogenous} for the complete derivation,
    including the feasible-set threshold.
\end{proof}

Because $V(\theta_c)$ is U-shaped with a unique interior minimum at
$\theta_c = \lambda$, no interior optimum exists; the firm's capability choice
necessarily collapses to a boundary solution, making the discrete
legacy-versus-frontier framing the mathematically exact representation of the
organizational choice. Since $V(\theta_c)=[\alpha^*(\theta_c)]^2/2$, the
deployment paradox carries directly into capability selection: within the gap
region, a frontier upgrade simultaneously reduces both optimal deployment and
firm value.

This result changes the interpretation of frontier AI adoption. The firm's
choice is no longer simply whether a more capable system is technologically
superior. It is whether the firm's current governance environment can support
that capability without eroding value. Under weak governance, the frontier
system may expand what the technology can do while still reducing the level at
which the firm can profitably deploy it. Capability choice therefore becomes
inseparable from the firm's current governance environment.

We define the \textit{sandboxing trap} as the situation in which a firm
rationally retains a legacy AI system rather than adopting a frontier upgrade
because the governance-capability gap makes the upgrade lower firm value than
the legacy option.

\begin{remark}
    The sandboxing trap arises because the firm's value function remains below
    the legacy benchmark even after the frontier system has moved beyond the
    local paradox boundary. Once $\theta_F>\lambda$, an additional capability
    increment again raises deployment at the margin, but the cumulative damage
    incurred over the paradox region can still leave total firm value below
    that of the legacy system. The resulting attachment to legacy capability
    is therefore structural rather than behavioral: it reflects the firm's
    current governance environment, not switching costs or
    inertia.\footnote{Note that this result characterizes a single firm's
    optimal capability choice absent competitive pressure. In markets where
    rivals close the gap and adopt frontier capability, competitive dynamics
    may induce firms to adopt frontier systems even when $V(\theta_F) <
    V(\theta_L)$ in the single-firm benchmark. Such dynamics could weaken
    rational rejection and increase aggregate cyber exposure, a possibility
    we discuss in Section~\ref{sec:discussion}.}
\end{remark}

\subsection{Social Welfare, Externalities, and the Expanded Paradox}

The preceding analysis takes the firm's perspective. AI breaches, however,
rarely stay inside the firm. A compromised AI agent with cross-system authority
can propagate damage to customers, partners, and critical
infrastructure---consequences the firm does not fully bear. This matters because
the paradox depends on the loss environment the firm takes into account when
choosing deployment. Once part of that loss falls on outsiders, the firm's
private assessment understates the social burden of deploying capable AI.

We capture these breach externalities with $e\ge 0$, so that total social
damage from a breach is $(1+e)\cdot p\cdot L$.

\medskip
\noindent\textbf{First-Best (Joint Control).}
To characterize the social benchmark, consider a social planner that jointly
chooses deployment and defense to maximize $W=\pi-e\cdot p\cdot L$. Note that
$\pi$ already accounts for the private loss $p\cdot L$, so subtracting
$e\cdot p\cdot L$ yields total social damage $(1+e)\cdot p\cdot L$---there is
no double-counting. This is equivalent to the firm's problem with $\lambda$
replaced by $(1+e)\lambda$.

\begin{proposition}[First-Best Social Optimum]
    \label{prop:social_A}
    \begin{equation}
        \alpha^{**}_{FB}(\theta) = \max\!\left\{0,\;
        \begin{cases}
            \mu+\theta(1-(1+e)\lambda), & (1+e)\lambda\theta\le 1, \\[4pt]
            \theta+\mu+1-2\sqrt{(1+e)\lambda\theta}, & (1+e)\lambda\theta>1.
        \end{cases}
        \right\}.
    \end{equation}
\end{proposition}

\begin{proof}
    The social planner's problem is equivalent to the firm's problem with
    $\lambda$ replaced by $(1+e)\lambda$ throughout. Applying
    Proposition~\ref{prop:alpha_star} with this substitution yields the stated
    expression. See Appendix~\ref{app:proof_social_A} for the complete
    derivation.
\end{proof}

This result sharpens the interpretation of the paradox. The private firm can
already deploy less than the no-risk benchmark in high-gap environments. Once
externalities are added, however, the same firm may still deploy more than is
socially optimal because it does not bear the full downside of a breach. The
relevant comparison therefore changes: the issue is no longer only how far
deployment falls short of the no-risk benchmark, but also how far private
deployment remains above the social optimum.

\medskip
\noindent\textbf{Second-Best (Deployment Regulation Only).}
To characterize the constrained social benchmark, consider a regulator that
sets deployment while the firm continues to choose $d$ privately.

\begin{proposition}[Second-Best Social Optimum]
    \label{prop:social_B}
    \begin{equation}
        \alpha^{**}_{SB}(\theta) =
        \max\!\left\{0,\;
        \begin{cases}
            \mu+\theta(1-(1+e)\lambda), & \lambda\theta\le 1, \\[4pt]
            \theta+\mu+1-(2+e)\sqrt{\lambda\theta}, & \lambda\theta>1.
        \end{cases}
        \right\}.
        \label{eq:social_B}
    \end{equation}
\end{proposition}

\begin{proof}
    The regulator anticipates $d^*(\alpha)$ from Lemma~\ref{lem:dstar} and
    substitutes into $W_{SB} \equiv \pi(\alpha,d^*) - e\cdot
    p(\alpha,d^*)\cdot L(\alpha,\theta)$. In each regime, $W_{SB}$ reduces
    to a concave quadratic in $\alpha$ (SOC $= -1 < 0$), and the FOC yields
    the stated expression. The regime threshold $\lambda\theta = 1$ is
    invariant to the regulator's choice of $\alpha$: the regulator scales the
    level of $d^*$ but cannot shift which regime the firm occupies. See 
    Appendix~\ref{app:proof_social_B} for the complete derivation.
\end{proof}

This second-best result arises because the regulator can restrict deployment
but cannot directly control the firm's defensive response. As a result, the
social adjustment must rely more heavily on deployment reduction.

\begin{proposition}[Externalities and the Social Paradox Region]
\label{prop:externality_expansion}
Let $e>0$ and suppose the private positive-deployment condition
$\lambda<\mu+1$ holds. Define $x\equiv\lambda\theta$ and $E\equiv 1+e$. Then:
\begin{enumerate}
    \item Social deployment is weakly below private deployment:
    \[
        \alpha^{**}_{SB}(\theta)\le
        \alpha^{**}_{FB}(\theta)\le
        \alpha^*(\theta).
    \]
    \item The first-best paradox region is weakly wider than the private
    paradox region. In the corner regime, the condition is
    $(1+e)\lambda>1$. In the interior regime, it is $\theta<(1+e)\lambda$.
    \item The second-best paradox region is weakly wider than the private
    paradox region. In the interior regime, it is
    $\theta<\left(\frac{2+e}{2}\right)^2\lambda$.
    \item There exist parameter values for which a firm is outside the private
    paradox region but inside the social paradox region.
\end{enumerate}
\end{proposition}

\begin{table}[htbp]
\centering
\caption{Paradox Conditions: Private vs.\ Social}
\label{tab:paradox_conditions}
\begin{tabular}{lcc}
\toprule
 & \textbf{Corner condition} & \textbf{Interior condition} \\
\midrule
Private       & $\lambda>1$       & $\theta<\lambda$ \\
First-Best    & $(1+e)\lambda>1$  & $\theta<(1+e)\lambda$ \\
Second-Best   & $(1+e)\lambda>1$  & $\theta<\left(\frac{2+e}{2}\right)^2\lambda$ \\
\bottomrule
\end{tabular}
\\\smallskip
\begin{minipage}{0.85\textwidth}
\small\textit{Note.} For $e>0$, social paradox regions weakly expand relative
to the private region. Each condition is interpreted within its corresponding
regime and subject to regime feasibility.
\end{minipage}
\end{table}

\begin{proof}
    See Appendix~\ref{app:proof_externality}.
\end{proof}

This result extends the paradox from a private deployment problem to a social
governance problem. Once external breach losses are taken into account, the
social planner assigns greater weight to the security burden of capable AI than
the private firm does. The relevant policy issue is therefore not simply
whether deployment should be restricted. It is whether policy can reduce the
loss environment that makes deployment socially costly in the first place. Deployment limits arise as a second-best response when the regulator cannot
directly improve or induce the firm's defensive and containment choices. Figure~\ref{fig:social_app} in Appendix~\ref{app:comp_statics} illustrates
how breach externalities widen the paradox region relative to the private
benchmark.

\section{Numerical Illustration}
\label{sec:numerical}

The numerical exercises illustrate how the model's comparative statics operate
under ordinal variation in breach-loss magnitude. They do not estimate
industry-specific AI deployment levels. The exercise has three purposes.
First, it maps industry breach-cost differences into relative values of
$\lambda$. Second, it distinguishes industries that admit a paradox region
from industries that are inside that region at a given capability level. Third,
it connects the continuous deployment result to the discrete capability-upgrade
decision in Proposition~\ref{prop:endogenous}. Direct measurement of
capability--authority co-variation is not currently feasible at scale: no
existing dataset records how broadly AI systems are granted operational
authority as a function of their capability. The numerical values should
therefore be read as internally consistent parameterizations of the model, not
as estimates of observed AI deployment.

\noindent\textbf{Ordinal calibration of breach-loss magnitude.} We set
$\theta=2$, $\mu=2$, and calibrate $\lambda$ using IBM's 2024 breach cost
data \citep{ibm2024}. We normalize the global average breach cost (\$4.88M)
to $\lambda=1$ and scale industries relative to that benchmark. We use the
three highest-cost industries in the IBM ranking that map naturally into our
setting---Healthcare, Financial Services, and Industrial---together with
Retail as a lower-cost benchmark (Table~\ref{tab:lambda}). This calibration
preserves relative breach-loss severity across industries. It does not treat
any dollar value as the economic source of the paradox.

\begin{table}[htbp]
\centering
\caption{Breach Loss Magnitude Parameters}
\label{tab:lambda}
\begin{tabular}{lcc}
\toprule
\textbf{Industry} & \textbf{IBM 2024 breach cost} & $\lambda$ \\
\midrule
Retail & \$3.48M & 0.71 \\
Industrial & \$5.56M & 1.14 \\
Financial Services & \$6.08M & 1.25 \\
Healthcare & \$9.77M & 2.00 \\
\bottomrule
\end{tabular}
\\\smallskip
\begin{minipage}{0.75\textwidth}
\small\textit{Note.} $\lambda$ is an ordinal index of conditional breach loss
magnitude under gap conditions.
\end{minipage}
\end{table}

The global average breach cost serves only as a scaling benchmark. The
numerical location of each threshold depends on the normalization of
productivity and loss scales. The substantive result is the existence of a
nonempty region in which the exposure-linked security burden rises faster than
the productivity gain, not the numerical value of any particular boundary.
Under this normalization, $\lambda=1$ marks whether an industry admits a
private paradox region somewhere in capability space. This condition is
distinct from whether the industry is currently inside that region at a given
capability level. The active paradox condition is $\theta<\lambda$. Thus,
$\lambda>1$ is necessary for a private paradox region to exist, but it is not
sufficient for the paradox to be active at every $\theta$.

At the illustrative capability level $\theta=2$, Retail ($\lambda=0.71$) does
not admit a private paradox region. Industrial ($\lambda=1.14$) and Financial
Services ($\lambda=1.25$) admit paradox regions, but they have already passed
their turning points and therefore exhibit no active paradox at $\theta=2$.
Healthcare ($\lambda=2.00$) lies at the turning point. The cross-industry
comparison therefore identifies which industries admit a paradox region and
which are active at the illustrative capability level, not which industries
are paradox-prone at every capability level.

\noindent\textbf{Endogenous capability choice.}
Table~\ref{tab:upgrade} illustrates Proposition~\ref{prop:endogenous} under
the same ordinal calibration, using $\mu=2$ and a discrete upgrade from a
legacy system ($\theta_L=0.5$) to a frontier system ($\theta_F=2$). The table
shows that the same comparative pattern carries over. Low-gap parameterizations
adopt the frontier upgrade because the productivity gain dominates the
additional security burden, whereas high-gap parameterizations reject it
because the security burden grows too strongly under current governance
conditions.

\begin{table}[H]
\centering
\caption{Capability Upgrade Decisions Under the Gap}
\label{tab:upgrade}
\begin{tabular}{lcccc}
\toprule
Industry & $\lambda$ & $V(\theta_L=0.5)$ & $V(\theta_F=2)$ & Upgrade? \\
\midrule
Retail ($\lambda=0.71$)     & 0.71 & 2.30 & 3.42 & Yes \\
Industrial ($\lambda=1.14$) & 1.14 & 1.86 & 1.96 & Yes \\
Financial ($\lambda=1.25$)  & 1.25 & 1.76 & 1.69 & No  \\
Healthcare ($\lambda=2.00$) & 2.00 & 1.13 & 0.50 & No  \\
\bottomrule
\end{tabular}
\\\smallskip
\begin{minipage}{0.8\textwidth}
\small\textit{Notes.} More severe gap parameterizations reject the frontier
upgrade, while milder ones adopt it. In highly regulated sectors such as
healthcare, the frontier system can amplify breach severity faster than it
raises productivity under current governance conditions.
\end{minipage}
\end{table}

\noindent\textbf{Baseline cross-industry comparison.}
Table~\ref{tab:numerical} reports the model's private-optimum predictions at
$\theta=2$ across the four industry parameterizations. At this capability level
all four industries lie in the interior regime ($\lambda\theta>1$). The table
also shows the economic magnitude of the security discount: although all
industries share the same no-risk benchmark $\alpha^0=4$, optimal deployment
varies sharply with $\lambda$, from 2.62 in Retail to 1.00 in Healthcare.

\begin{table}[H]
\centering
\caption{Model Predictions Across Industries ($\theta=2$, $\mu=2$)}
\label{tab:numerical}
\begin{tabular}{lccccccc}
\toprule
\textbf{Industry} & $\lambda$ & $\alpha^0$ & $\alpha^*$ & $d^*$
    & $p^*$ & \shortstack{Admits paradox\\region?}
    & \shortstack{Active paradox\\at $\theta=2$?} \\
\midrule
Retail            & 0.71 & 4.00 & 2.62 & 0.50 & 0.84 & No  & No \\
Industrial        & 1.14 & 4.00 & 1.98 & 1.00 & 0.66 & Yes & No \\
Financial Services& 1.25 & 4.00 & 1.84 & 1.07 & 0.63 & Yes & No \\
Healthcare        & 2.00 & 4.00 & 1.00 & 1.00 & 0.50 & Yes & Boundary \\
\bottomrule
\end{tabular}
\\\smallskip
\begin{minipage}{0.9\textwidth}
\small\textit{Note.} $\alpha^0=\theta+\mu=4$ (no-risk benchmark). ``Admits
paradox region'' indicates whether $\lambda>1$. ``Active paradox at $\theta=2$''
indicates whether $\partial\alpha^*/\partial\theta<0$ at the illustrated
capability level.
\end{minipage}
\end{table}

\noindent\textbf{Tracing the paradox over capability levels.}
Table~\ref{tab:dynamics} traces optimal deployment over a range of capability
levels. The paradox is not a level effect but a slope effect: in industries
with $\lambda>1$, deployment initially falls as capability rises, then bottoms
out at $\theta=\lambda$, and rises thereafter. Bold entries mark the active
paradox region ($\theta<\lambda$); daggers mark the corner regime
($\lambda\theta\le 1$).

\begin{table}[H]
\centering
\caption{Deployment Paradox: $\alpha^*$ as $\theta$ Increases ($\mu=2$)}
\label{tab:dynamics}
\begin{tabular}{lccccccc}
\toprule
& $\theta\!=\!0.5$ & $\theta\!=\!1$ & $\theta\!=\!1.5$ & $\theta\!=\!2$
    & $\theta\!=\!3$ & $\theta\!=\!4$ & $\theta\!=\!5$ \\
\midrule
Retail ($\lambda\!=\!0.71$)
    & 2.15$^\dag$ & 2.29$^\dag$ & 2.44 & 2.62 & 3.08 & 3.63 & 4.23 \\
Industrial ($\lambda\!=\!1.14$)
    & \textbf{1.93}$^\dag$ & \textbf{1.87} & 1.88 & 1.98
    & 2.30 & 2.73 & 3.22 \\
Financial ($\lambda\!=\!1.25$)
    & \textbf{1.88}$^\dag$ & \textbf{1.76} & 1.76 & 1.84
    & 2.13 & 2.53 & 3.00 \\
Healthcare ($\lambda\!=\!2.00$)
    & \textbf{1.50}$^\dag$ & \textbf{1.17} & \textbf{1.04} & 1.00
    & 1.10 & 1.34 & 1.68 \\
\bottomrule
\end{tabular}
\\\smallskip
\begin{minipage}{0.75\textwidth}
\small\textit{Note.} Bold: $\theta<\lambda$ (active paradox region). $\dag$:
Corner regime ($\lambda\theta\le 1$). Turning point at $\theta=\lambda$.
\end{minipage}
\end{table}

\noindent\textbf{Visualizing the U-shaped deployment paths.}
Figure~\ref{fig:paradox} shows the graphical counterpart to
Table~\ref{tab:dynamics}. The U-shape arises because the security burden
dominates early in the capability trajectory but attenuates once capability
advances far enough relative to $\lambda$.

\begin{figure}[H]
    \centering
    \includegraphics[width=0.68\textwidth]{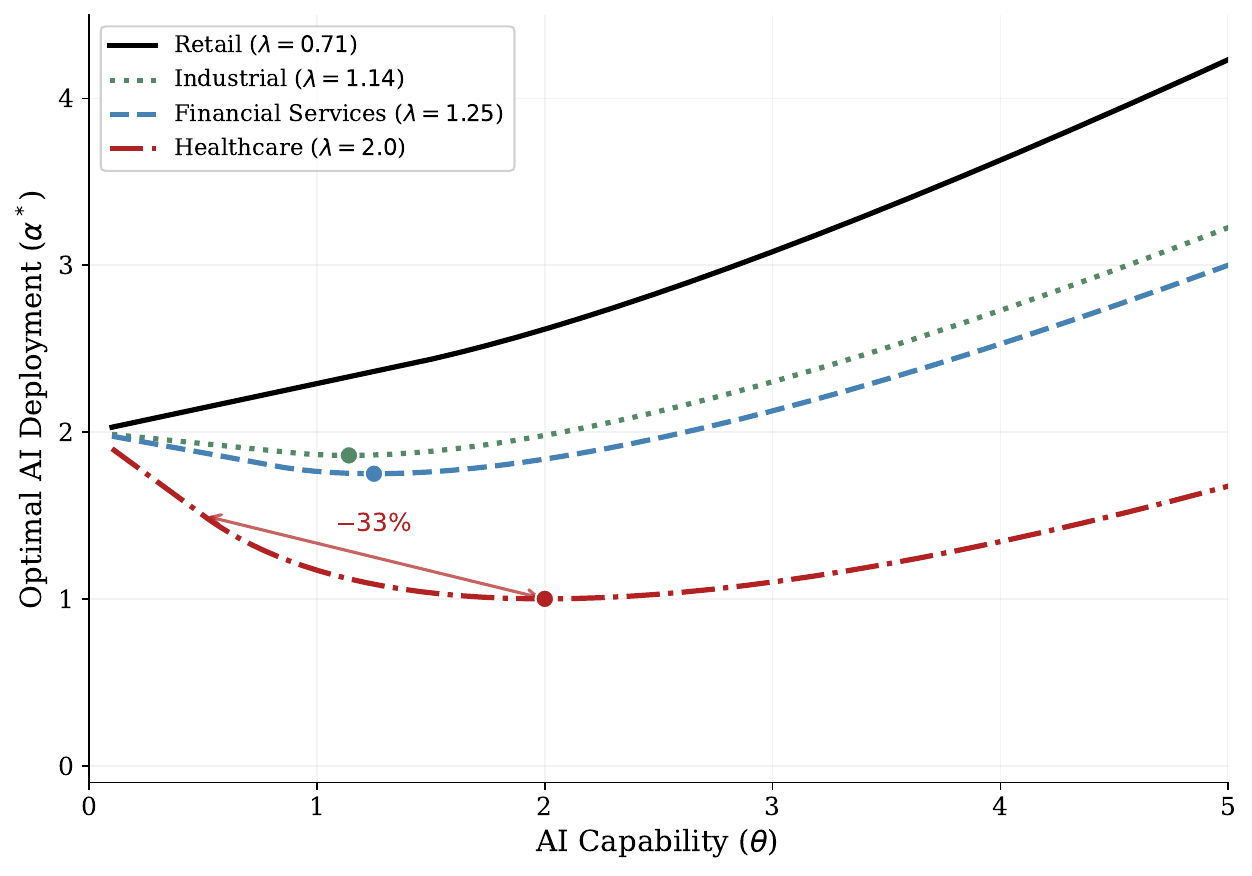}
    \caption{The deployment paradox under the gap. Retail ($\lambda=0.71$)
    increases monotonically. Industries with $\lambda>1$ exhibit a U-shaped
    pattern; turning points are at $\theta=\lambda$ (dots). When governance
    investment reduces $\lambda$ below 1, the paradox is eliminated.}
    \label{fig:paradox}
\end{figure}

\noindent\textbf{From private to social deployment.}
Table~\ref{tab:externality} shows how the same industries are reclassified
once breach externalities are introduced under the Second-Best case.

\begin{table}[H]
\centering
\caption{Social Welfare Predictions ($\theta=2$, $\mu=2$)}
\label{tab:externality}
\begin{tabular}{lccccc}
\toprule
\textbf{Industry} & $\lambda$ & $e$ & $\alpha^*$ & $\alpha^{**}_{SB}$
    & \shortstack{Has paradox\\region? (social)} \\
\midrule
Retail             & 0.71 & 0.3 & 2.62 & 2.26 & No  \\
Industrial         & 1.14 & 0.5 & 1.98 & 1.22 & Yes \\
Financial Services & 1.25 & 1.0 & 1.84 & 0.26 & Yes \\
Healthcare         & 2.00 & 1.5 & 1.00 & 0    & Yes \\
\bottomrule
\end{tabular}
\\\smallskip
\begin{minipage}{0.85\textwidth}
\small\textit{Note.}
$\alpha^{**}_{SB}=\max\{0,\;\theta+\mu+1-(2+e)\sqrt{\lambda\theta}\}$.
The social paradox region is strictly wider than the private one.
\end{minipage}
\end{table}

\section{Robustness Checks}

This section examines whether the deployment paradox depends on the baseline
functional forms. The baseline model uses a Tullock contest function
$(\beta=1)$, a proportional authority-exposure specification $a(\theta)=\theta$
that implies $L=\lambda\alpha\theta$, and a linear capability-productivity term
$\theta\alpha$. The extensions vary one component at a time to show that the
mechanism does not rely on a knife-edge specification.

\subsection{Role of Complementary Organizational Readiness}
\label{subsec:mu_invariance}

Complementary organizational readiness shifts deployment levels but does not
drive the paradox boundary in the baseline model. In
Proposition~\ref{prop:paradox}, the sign of $\partial\alpha^*/\partial\theta$
depends on $\lambda$ and $\theta$, not on $\mu$. Thus higher $\mu$ raises the
level of $\alpha^*$ and firm value, but it does not change the turning point
$\theta=\lambda$ under the baseline specification.
Figure~\ref{fig:mu_robust} illustrates this level shift for $\lambda=1.25$
across four values of $\mu$.

Allowing $\mu$ to partially enter the damage channel changes this boundary but
does not eliminate the mechanism. Appendix~\ref{app:mu_damage} gives the
corresponding condition.

\begin{figure}[htbp]
    \centering
    \includegraphics[width=0.63\textwidth]{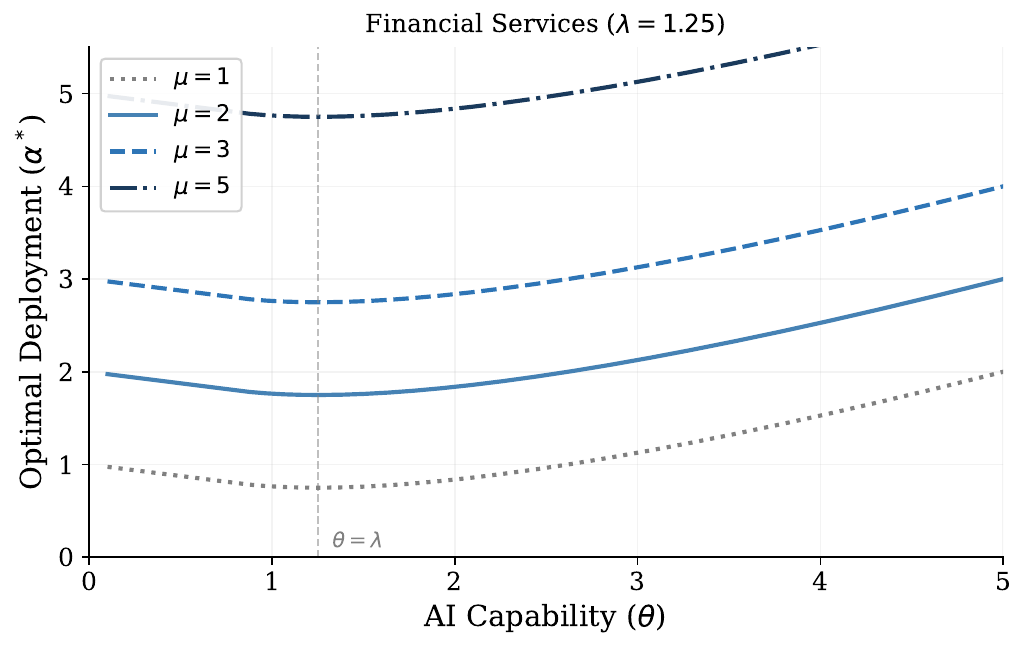}
    \caption{Robustness to $\mu$ ($\lambda=1.25$). $\alpha^*$ vs.\ $\theta$
    for $\mu\in\{1,2,3,5\}$. Higher $\mu$ shifts deployment upward. The
    turning point remains $\theta=\lambda=1.25$.}
    \label{fig:mu_robust}
\end{figure}

\subsection{Generalized Authority Exposure and Damage}
\label{subsec:gen_damage}

Appendix~\ref{app:generalized_L} generalizes the authority-exposure index to
$a(\theta)=\theta^\gamma$, $\gamma\ge 0$. When $\gamma=0$, capability is fully
decoupled from operational exposure, and the deployment paradox disappears.
When $\gamma>0$, sign reversal can arise over a nonempty parameter region
whenever the exposure-induced increase in the security burden exceeds the
marginal productivity gain. The baseline case $\gamma=1$ represents one
transparent specification of a more general authority-exposure mechanism.

\subsection{Generalized Breach Probability}
\label{subsec:gen_prob}

Appendix~\ref{app:generalized_p} generalizes the breach probability to
$p(\alpha,d)=\alpha^\beta/(\alpha^\beta+d)$, $\beta\ge 1$. The baseline case
$\beta=1$ is conservative because optimal defense scales proportionally with
deployment. When $\beta>1$, the attack-surface channel can reinforce the
baseline damage-channel mechanism locally under interior feasibility.

\subsection{Generalized Productivity Function}
\label{subsec:gen_prod}

Appendix~\ref{app:generalized_prod} replaces the baseline productivity term
$\theta\alpha$ with $\theta^\eta\alpha$, where $\eta>0$. Sign reversal arises
when the marginal security burden from capability exceeds the marginal
productivity gain. The condition changes with $\eta$, but the economic logic
does not. Stronger capability-productivity elasticity narrows the paradox
region; stronger capability-damage exposure widens it.

\subsection{Parameter Space}
\label{subsec:param_space}

Figure~\ref{fig:heatmap} maps optimal deployment $\alpha^*$ across the
$(\theta,\lambda)$ space. The paradox region appears over a broad range, not
a knife-edge point.

\begin{figure}[H]
    \centering
    \includegraphics[width=0.63\textwidth]{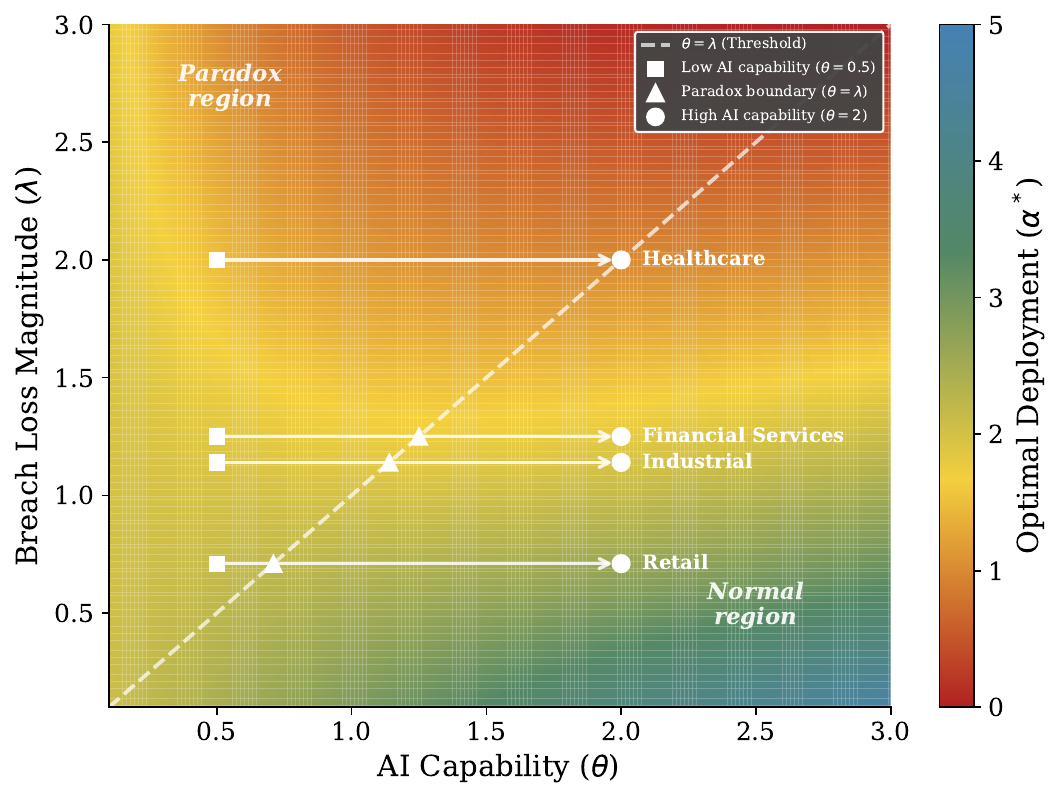}
    \caption{Industry regimes in $(\theta,\lambda)$ space. Dashed diagonal:
    paradox boundary $\theta=\lambda$. Arrows trace illustrative trajectories
    under the ordinal calibration of Table~\ref{tab:lambda}.}
    \label{fig:heatmap}
\end{figure}

\section{Discussion}
\label{sec:discussion}

\subsection{Industry Heterogeneity and Testable Predictions}

\begin{corollary}[Industry Sorting]
    \label{cor:industry}
    Industries sort into three regimes. First, when $\lambda \le 1$, no
    deployment paradox arises. Second, when $\lambda > 1$ and $\theta <
    \lambda$, the firm lies in the paradox region and optimal deployment
    decreases with capability. Third, when $\lambda > 1$ and $\theta \ge
    \lambda$, the paradox no longer holds and optimal deployment is weakly
    increasing in capability.
\end{corollary}

Standard capability-driven models predict that better AI weakly increases
deployment across industries. Our model instead predicts a sign reversal that
depends on the firm's conditional breach-loss environment. This yields a
natural difference-in-slopes test. One can regress industry- or firm-level AI
deployment intensity on capability advances, interacted with proxies for
conditional breach loss magnitude such as breach-cost ratios, downstream
propagation potential, or regulatory-liability measures. The model predicts a
negative interaction in the paradox region and a positive interaction outside
it.

A related empirical challenge concerns the operationalization of capability and
authority. In the model, $\theta$ captures capability, while the governance gap
links capability to operational exposure. Natural proxies can include system
autonomy, data-access breadth, task consequentiality, and authority
scope---dimensions partly reflected in risk-classification frameworks such as
the NIST AI Risk Management Framework and the EU AI Act \citep{nist2023,
euaiact2024}.

The model also shows heterogeneity among firms with access to similar frontier
systems. Once breach externalities are taken into account, private deployment
can exceed the social optimum even in sectors that lie outside the private
paradox region. These implications suggest three empirical patterns: capability
improvements should increase deployment less, or reduce deployment, in
high-loss sectors; this response should weaken among firms with stronger
governance controls; and weak-governance firms may retain legacy systems or
restrict frontier systems to sandboxed use cases. Testing these patterns
requires separate measures of AI capability, authority exposure, and governance
maturity. The prediction is falsifiable. If high-loss firms with weak authority
decoupling show monotone increases in AI deployment as capability improves,
with no reduction or flattening over the predicted range, the deployment
paradox would not describe that setting. Empirically, the relevant test would
combine variation in breach-loss environments with firm-level measures of
authority exposure or governance maturity, such as access-control maturity,
segmentation quality, or privilege-audit outcomes.

\subsection{Sources of the Governance-Capability Gap}

The model does not require a unique micro-foundation for why the
governance-capability gap arises. Several organizational frictions can create
this linkage. Firms may find it costly to redesign segmentation and privilege
boundaries. Operating units may resist access restrictions that slow workflow
integration. Managers may also find it difficult to audit the authority granted
to rapidly evolving AI systems. These sources differ, but they have the same
implication for the model: they allow capability to expand faster than
governance controls can contain it.

\subsection{Managerial and Policy Implications}

\paragraph{Managerial implications.}
The model separates productivity-side investments from security-side
investments. Workflow redesign, data governance, process standardization, and
organizational learning raise deployment value, but they need not reduce
conditional breach-loss magnitude. By contrast, containment quality, data
segmentation, and least-privilege architecture lower the loss environment, but
they need not create productivity gains by themselves.

This distinction gives managers a diagnostic use for the security discount. A
large discount indicates that governance maturity is constraining AI value. In
such environments, reducing breach-loss magnitude may increase economically
viable deployment more than further capability expansion.

The same logic applies to capability upgrades. A firm in the sandboxing
trap---where frontier deployment is value-reducing under current governance
conditions---may gain more from reducing breach-loss magnitude than from
adopting a more capable system. Governance investment that precedes capability
advancement expands the range of productive deployment.

\paragraph{Policy implications.}
The model does not imply that deployment restrictions are the preferred policy
instrument. Deployment-only regulation appears as a second-best response when
policy cannot directly change the firm's defense and containment environment.

Policies that reduce conditional breach damage target the mechanism that
creates the paradox. Segmentation requirements, least-privilege rules,
containment standards, controls on lateral movement, and auditable authority
boundaries reduce the damage that follows from compromise. These policies
differ from generic security-spending mandates. Spending mandates can increase
current defensive effort, which corresponds to $d$ in the model. They do not
necessarily reduce post-breach propagation, which corresponds to $\lambda$.

Deployment caps, use-based restrictions, or high-risk deployment tiers are
therefore fallback instruments. A more direct policy design pairs broader
deployment permissions with governance conditions, such as data segmentation,
least-privilege access, incident containment, and authority logging.

\subsection{Limitations and Future Research}

Our study has limitations, which also suggest opportunities for future research. First, the model is static. We treat $\lambda$ as given within a single-period deployment problem and solve for optimal deployment at a contemporaneous capability level $\theta$. This abstraction does not capture the joint evolution of capability, defense, and governance over time, but it isolates the cross-sectional tradeoff a firm faces at a given capability-governance configuration. Appendix~\ref{app:gov_adjustment} sketches the broader dynamic logic by showing how governance adjustment and capability growth jointly shift the paradox boundary. Appendix~\ref{app:endogenous_gov} shows that a residual gap can remain privately rational when governance restructuring is costly. A full dynamic extension could endogenize governance investment, allow firms to trade off current deployment against future reductions in $\lambda$, and characterize repeated re-optimization as capability and governance co-evolve.

Second, the baseline abstracts from several within-firm complexities. Security investment $d$ is modeled as generic defense and therefore does not capture the possibility that broader AI deployment itself improves monitoring or response. Relatedly, the model does not allow AI capability to improve governance efficiency---for instance, through AI-powered anomaly detection, automated access-control enforcement, or continuous compliance monitoring---which would create a countervailing channel in which higher $\theta$ partially reduces $\lambda$. These omitted channels matter because they determine whether capability only amplifies exposure or can also help contain it. Richer models could therefore examine defensive-AI complementarities and AI-assisted governance maturation more explicitly.

Third, the attacker is passive. The model therefore captures how the firm's own deployment and defense choices shape cyber risk, but not how adversaries strategically respond to more capable or more broadly deployed AI. Appendix~\ref{app:strategic_attacker} sketches this direction. If attacker capability also scaled with $\theta$, the paradox region could widen because more capable AI would become a more attractive and potentially more exploitable target. A fuller extension would model strategic interaction between firm and attacker directly.

Fourth, $\alpha$ is a single deployment index, whereas firms in practice allocate AI across applications with different authority, exposure, and breach-loss profiles. Appendix~\ref{app:heterogeneous_apps} sketches a heterogeneous-applications extension in which firms allocate deployment across high-authority and low-authority uses. Such a framework could show how firms substitute toward lower-authority applications even when aggregate deployment declines under the governance-capability gap.

Finally, the baseline model focuses on the firm's deployment problem in isolation and therefore abstracts from strategic interaction across firms. An industry-equilibrium extension could examine whether competitive pressure leads firms to collectively over-deploy into the paradox region, eroding industry-wide value while leaving aggregate cyber exposure elevated. Such a setting would further clarify when sector-level governance standards are socially valuable.

\section{Conclusion}

This paper identifies a condition under which stronger AI capability can reduce
rather than expand economically viable deployment. When governance lags
capability, more capable systems can carry greater authority exposure, causing
the security burden of deployment to rise faster than the productivity gain.
The model and its results have clear boundaries: they do not capture dynamic
learning, strategic attackers, heterogeneous applications, or competitive
pressure across firms. Despite these limitations, the paper shows that the
value of AI capability depends on the organizational conditions under which
firms deploy it. When governance decouples capability from authority exposure,
capability improvements again support broader deployment.

\newpage
\bibliographystyle{apalike}
\bibliography{references_IS}

\newpage
\appendix
\setcounter{figure}{0}
\setcounter{table}{0}
\renewcommand{\thefigure}{A\arabic{figure}}
\renewcommand{\thetable}{A\arabic{table}}

\section{Global Solution and Boundary Behavior}
\label{app:global}

Define $f_C(\theta)\equiv\mu+\theta(1-\lambda)$ and
$f_I(\theta)\equiv\theta+\mu+1-2\sqrt{\lambda\theta}$.
Under Assumption~\ref{assump:interior}, $\alpha^*=f_C$ in the corner regime
and $\alpha^*=f_I$ in the interior regime.

At $\theta=1/\lambda$: $f_C=f_I=\mu+1/\lambda-1$ and $f_C'=f_I'=1-\lambda$
($C^1$ continuity). In the interior regime: $f_I''=\frac{1}{2}\sqrt{\lambda/\theta^3}>0$;
minimum at $\theta=\lambda$ with $\alpha^*_{\min}=\mu+1-\lambda>0$. When
$\lambda>1$ the U-shape spans both regimes; when $\lambda\le 1$, $\alpha^*$
is increasing throughout. $\mu$ enters only as an additive constant; all
comparative statics are invariant to $\mu$.

\paragraph{Positivity across regimes.}
Under Assumption~\ref{assump:interior}, $\alpha^*(\theta)>0$ for all
$\theta>0$. When $\lambda\le 1$: $\alpha^*$ is increasing throughout and
$\alpha^*(0)=\mu>0$. When $\lambda>1$: $\alpha^*_{\min}=\mu+1-\lambda>0$ at
the interior minimum $\theta=\lambda$. The corner regime bounds satisfy
$\alpha^*(0)=\mu>0$ and
$\alpha^*(1/\lambda)=\mu+1/\lambda-1>\mu+1-\lambda>0$ by AM--GM applied to
$\lambda+1/\lambda\ge 2$. Hence $\alpha^*>0$ globally.

\section{Proof of Lemma~\ref{lem:dstar}: Optimal Security Investment}
\label{app:proof_dstar}

\begin{proof}
For a given deployment level $\alpha>0$, differentiate the profit function
with respect to $d$:
\[
\frac{\partial \pi}{\partial d}
=
\frac{\lambda \alpha^2 \theta}{(\alpha+d)^2}-1.
\]
Evaluating at the boundary $d=0$ gives
\[
\left.\frac{\partial \pi}{\partial d}\right|_{d=0}
=
\lambda\theta-1.
\]
Thus, when $\lambda\theta\le 1$, the marginal benefit of AI-specific defense
never exceeds its marginal cost, so the non-negativity constraint binds and
$d^*=0$. Economically, this is the case in which the combined effect of
capability and breach-loss magnitude is too small to justify dedicated
defensive expenditure.

When $\lambda\theta>1$, setting the first-order condition equal to zero yields
\[
\frac{\lambda \alpha^2 \theta}{(\alpha+d)^2}=1
\quad\Longrightarrow\quad
\alpha+d=\alpha\sqrt{\lambda\theta}
\quad\Longrightarrow\quad
d=\alpha(\sqrt{\lambda\theta}-1).
\]
Because
\[
\frac{\partial^2 \pi}{\partial d^2}
=
-\frac{2\lambda\alpha^2\theta}{(\alpha+d)^3}<0,
\]
the objective is strictly concave in $d$, so this solution is the unique
global maximum. Therefore,
\[
d^*(\alpha,\theta)=
\begin{cases}
0, & \lambda\theta\le 1,\\[4pt]
\alpha(\sqrt{\lambda\theta}-1), & \lambda\theta>1.
\end{cases}
\]
This expression has a direct operational reading: once defense becomes
worthwhile, it scales in direct proportion to deployment $\alpha$, and the
per-unit defensive burden $\sqrt{\lambda\theta}-1$ rises with both
breach-loss magnitude and capability.
\end{proof}

\section{Proof of Lemma~\ref{lem:pstar}: Equilibrium Breach Probability}
\label{app:proof_pstar}

\begin{proof}
By Lemma~\ref{lem:dstar}, the firm's optimal defense rule is already available
in closed form. This proof simply substitutes that rule into the
breach-probability function
\[
p(\alpha,d)=\frac{\alpha}{\alpha+d}.
\]

If $\lambda\theta\le 1$, then $d^*=0$, so
\[
p^*=\frac{\alpha}{\alpha+0}=1.
\]
At this corner, the AI-attributable incremental exposure margin is left fully
undefended.

If $\lambda\theta>1$, then
$d^*=\alpha(\sqrt{\lambda\theta}-1)$, so
\[
p^*
=
\frac{\alpha}{\alpha+\alpha(\sqrt{\lambda\theta}-1)}
=
\frac{\alpha}{\alpha\sqrt{\lambda\theta}}
=
\frac{1}{\sqrt{\lambda\theta}}.
\]
In both regimes, $\alpha$ cancels out, so
\[
p^*(\theta)=
\begin{cases}
1, & \lambda\theta\le 1,\\[4pt]
\dfrac{1}{\sqrt{\lambda\theta}}, & \lambda\theta>1.
\end{cases}
\]
This cancellation is economically important. Under the baseline Tullock form, a
broader AI footprint does not mechanically raise equilibrium breach
probability because a rational firm scales defense in step with deployment.
Deployment therefore changes the cost of maintaining security, not the
equilibrium breach probability itself.
\end{proof}

\section{Proof of Lemma~\ref{lem:reduced}: Reduced-Form Profit}
\label{app:proof_lemma3}

\begin{proof}
Lemma~\ref{lem:dstar} gives the firm's optimal defense rule, and
Lemma~\ref{lem:pstar} gives the associated equilibrium breach probability.
This proof combines those two objects and substitutes them into
\eqref{eq:profit} to obtain profit as a function of deployment alone.

\medskip
\noindent\textbf{Corner regime ($\lambda\theta\le 1$).}
From Lemmas~\ref{lem:dstar}--\ref{lem:pstar}, we have $d^*=0$ and $p^*=1$.
Hence expected security loss is
\[
p^*\cdot L = 1\cdot (\lambda\alpha\theta)=\lambda\alpha\theta.
\]
Substituting into \eqref{eq:profit} gives
\[
\pi(\alpha;\theta)
=
(\theta+\mu)\alpha-\frac{\alpha^2}{2}-\lambda\alpha\theta
=
\bigl[\mu+\theta(1-\lambda)\bigr]\alpha-\frac{\alpha^2}{2}.
\]
Intuitively, capability still raises productivity in this regime, but the firm
undertakes no additional defense, so expected breach losses reduce the net
marginal value of capability to $1-\lambda$.

\medskip
\noindent\textbf{Interior regime ($\lambda\theta>1$).}
From Lemmas~\ref{lem:dstar}--\ref{lem:pstar},
\[
d^*=\alpha(\sqrt{\lambda\theta}-1),
\qquad
p^*=\frac{1}{\sqrt{\lambda\theta}}.
\]
Expected security loss becomes
\[
p^*\cdot L
=
\frac{1}{\sqrt{\lambda\theta}}\cdot (\lambda\alpha\theta)
=
\alpha\sqrt{\lambda\theta}.
\]
Substituting into \eqref{eq:profit},
\begin{align*}
\pi(\alpha;\theta)
&=
(\theta+\mu)\alpha-\frac{\alpha^2}{2}
-\alpha\sqrt{\lambda\theta}
-\alpha(\sqrt{\lambda\theta}-1)\\
&=
(\theta+\mu)\alpha-\frac{\alpha^2}{2}
-2\alpha\sqrt{\lambda\theta}+\alpha\\
&=
\bigl[\theta+\mu+1-2\sqrt{\lambda\theta}\bigr]\alpha
-\frac{\alpha^2}{2}.
\end{align*}
Intuitively, in this regime the firm bears the security cost of deployment
through two channels: residual breach damage and defensive expenditure.
Together, these imply a total security-related deduction of
$(2\sqrt{\lambda\theta}-1)\alpha$ in reduced form.
Equivalently, the linear coefficient on deployment is
$\theta+\mu+1-2\sqrt{\lambda\theta}$, so its derivative with respect to
capability is $1-\sqrt{\lambda/\theta}$.
When $\theta<\lambda$, the security burden rises faster than the productivity
gain, generating the deployment paradox.

Therefore,
\[
\pi(\alpha;\theta)=
\begin{cases}
\bigl[\mu+\theta(1-\lambda)\bigr]\alpha-\dfrac{\alpha^2}{2},
& \lambda\theta\le 1,\\[8pt]
\bigl[\theta+\mu+1-2\sqrt{\lambda\theta}\bigr]\alpha-\dfrac{\alpha^2}{2},
& \lambda\theta>1.
\end{cases}
\]
The reduced-form profit makes the tradeoff transparent:
higher capability increases gross value, but under the governance gap it also increases the security-related cost of deployment.
\end{proof}

\section{Proof of Proposition~\ref{prop:alpha_star}: Optimal AI Deployment}
\label{app:proof_alpha_star}

\begin{proof}
Lemma~\ref{lem:reduced} shows that once the firm chooses defensive spending
optimally within a given regime, its remaining decision is how broadly to
deploy AI ($\alpha$). In each regime, the firm therefore chooses $\alpha$ by
weighing the marginal productivity gain from broader deployment against the
marginal security burden implied by that regime. This yields a single
first-order condition for optimal deployment in each regime, and the two
regime-specific expressions connect smoothly at the boundary.

\medskip
\noindent\textbf{Corner regime ($\lambda\theta\le 1$).}
From Lemma~\ref{lem:reduced},
\[
\pi(\alpha;\theta)=\bigl[\mu+\theta(1-\lambda)\bigr]\alpha-\frac{\alpha^2}{2}.
\]
The first-order condition is
\[
\frac{\partial\pi}{\partial\alpha}
=\mu+\theta(1-\lambda)-\alpha=0
\quad\Longrightarrow\quad
\alpha^*(\theta)=\mu+\theta(1-\lambda).
\]
Because $\partial^2\pi/\partial\alpha^2=-1<0$, the objective is strictly
concave in $\alpha$, so this is the unique global maximum. Economically, this expression has a clear organizational interpretation. In the
corner regime, the firm does not find it worthwhile to invest in AI-specific
defense, so its remaining choice is how broadly to deploy AI ($\alpha$) given
the expected breach burden it must bear. Complementary organizational readiness, captured by $\mu$, always
pushes deployment upward because it raises the value of broader adoption
without amplifying breach damage. Capability $\theta$, by contrast, expands
deployment only when the loss environment is sufficiently mild. Its marginal
effect in this regime is $1-\lambda$: when $\lambda<1$, the productivity gain
from greater capability dominates and the firm broadens deployment; when
$\lambda>1$, greater capability instead makes deployment less attractive even
before defensive investment becomes worthwhile, because each increment of
capability raises expected breach loss more than it raises productivity. Under
Assumption~\ref{assump:interior} ($\lambda<\mu+1$), $\alpha^*(\theta)>0$
throughout the corner region.

\medskip
\noindent\textbf{Interior regime ($\lambda\theta>1$).}
From Lemma~\ref{lem:reduced},
\[
\pi(\alpha;\theta)=\bigl[\theta+\mu+1-2\sqrt{\lambda\theta}\bigr]\alpha
-\frac{\alpha^2}{2}.
\]
The first-order condition is
\[
\frac{\partial\pi}{\partial\alpha}
=\theta+\mu+1-2\sqrt{\lambda\theta}-\alpha=0
\quad\Longrightarrow\quad
\alpha^*(\theta)=\theta+\mu+1-2\sqrt{\lambda\theta}.
\]
Again $\partial^2\pi/\partial\alpha^2=-1<0$, so this is the unique global
maximum. This expression has a direct operational reading: the firm's
productive capacity $\theta+\mu+1$ is reduced by a security burden
$2\sqrt{\lambda\theta}$, capturing both residual breach damage and defensive
expenditure once AI-specific defense becomes worthwhile. The $+1$ term reflects the firm's ability to offset part of the breach-loss
burden through defense: without defense, that burden would be $\lambda\theta$
rather than $2\sqrt{\lambda\theta}-1$. Under Assumption~\ref{assump:interior},
$\alpha^*_{\min}=\mu+1-\lambda>0$ (Appendix~\ref{app:global}), so interior
deployment remains strictly positive even at the paradox trough.

\medskip
\noindent\textbf{Regime continuity at $\theta=1/\lambda$.}
At the boundary, the corner expression gives
\[
\alpha^*=\mu+\frac{1-\lambda}{\lambda}=\mu+\frac{1}{\lambda}-1,
\]
while the interior expression gives
\[
\alpha^*=\frac{1}{\lambda}+\mu+1-2\sqrt{\lambda\cdot\frac{1}{\lambda}}
=\frac{1}{\lambda}+\mu+1-2
=\mu+\frac{1}{\lambda}-1.
\]
The two expressions coincide, confirming that optimal deployment is continuous
across regimes. Thus, the regime shift at $\theta=1/\lambda$ appears not as a
jump in deployment, but as a change in slope: once AI-specific defense becomes
worthwhile, the firm's response function continues smoothly but becomes
steeper or flatter depending on the regime.
\end{proof}

\section{Proof of Proposition~\ref{prop:discount}: Security Discount}
\label{app:proof_discount}

\begin{proof}
Proposition~\ref{prop:alpha_star} already characterizes $\alpha^*$ in closed
form in both regimes, so the security discount $\Delta\equiv\alpha^0-\alpha^*$
follows by direct subtraction from the no-risk benchmark $\alpha^0=\theta+\mu$.

\medskip
\noindent\textbf{Corner regime ($\lambda\theta\le 1$).}
From Proposition~\ref{prop:alpha_star},
$\alpha^*=\mu+\theta(1-\lambda)$. Therefore,
\[
\Delta(\theta)
=(\theta+\mu)-\bigl[\mu+\theta(1-\lambda)\bigr]
=\theta-\theta(1-\lambda)
=\lambda\theta.
\]
Economically, in this regime the firm makes no AI-specific defensive
investment, so the entire expected breach damage $\lambda\alpha\theta$ shows
up as a reduction in deployment relative to the no-risk benchmark. The
discount therefore rises linearly with both capability and breach-loss
magnitude.

\medskip
\noindent\textbf{Interior regime ($\lambda\theta>1$).}
From Proposition~\ref{prop:alpha_star},
$\alpha^*=\theta+\mu+1-2\sqrt{\lambda\theta}$. Therefore,
\[
\Delta(\theta)
=(\theta+\mu)-\bigl[\theta+\mu+1-2\sqrt{\lambda\theta}\bigr]
=2\sqrt{\lambda\theta}-1.
\]
Here the firm is now paying twice: once for residual breach damage after
defense, and once for the defensive expenditure itself. That is why the
discount takes the square-root form $2\sqrt{\lambda\theta}-1$ rather than the
linear form $\lambda\theta$ of the corner regime---defensive investment
blunts the damage channel but introduces its own cost.

\medskip
\noindent\textbf{Continuity and smoothness at $\theta=1/\lambda$.}
At the boundary, the corner expression yields
$\Delta=\lambda\cdot(1/\lambda)=1$, while the interior expression yields
$\Delta=2\sqrt{\lambda\cdot(1/\lambda)}-1=2-1=1$. The two values coincide.

For the first derivative, the corner expression gives
$\Delta'(\theta)=\lambda$, while the interior expression gives
\[
\Delta'(\theta)=\frac{\partial}{\partial\theta}\bigl(2\sqrt{\lambda\theta}-1\bigr)
=\sqrt{\frac{\lambda}{\theta}}.
\]
At $\theta=1/\lambda$, the interior derivative equals
$\sqrt{\lambda\cdot\lambda}=\lambda$, matching the corner derivative. So
$\Delta$ is continuously differentiable across the regime boundary, even
though its functional form changes. This matters economically: the security
discount does not jump when the firm starts investing in defense; it only
shifts from linear to concave growth, reflecting the firm's ability to
partially absorb the damage through defensive expenditure.

\medskip
\noindent\textbf{Monotonicity and $\mu$-independence.}
In the corner regime, $\partial\Delta/\partial\lambda=\theta>0$ and
$\partial\Delta/\partial\theta=\lambda>0$. In the interior regime,
$\partial\Delta/\partial\lambda=\sqrt{\theta/\lambda}>0$ and
$\partial\Delta/\partial\theta=\sqrt{\lambda/\theta}>0$. In both regimes, the
discount grows with $\lambda$ and $\theta$, and is independent of $\mu$
because $\mu$ enters $\alpha^0$ and $\alpha^*$ additively and cancels in the
difference. The security discount therefore measures the pure gap-induced
wedge between productive and realized deployment, isolated from complementary organizational-readiness effects.
\end{proof}

\section{Proof of Proposition~\ref{prop:paradox}: Deployment Paradox}
\label{app:proof_paradox}

\begin{proof}
By Proposition~\ref{prop:alpha_star}, optimal deployment is already available
in closed form in both regimes. The proof therefore reduces to differentiating
$\alpha^*(\theta)$ with respect to capability and studying the sign of the
derivative.

\medskip
\noindent\textbf{Corner regime ($\lambda\theta\le 1$).}
From Proposition~\ref{prop:alpha_star},
\[
\alpha^*(\theta)=\mu+\theta(1-\lambda),
\]
so
\[
\frac{\partial \alpha^*}{\partial \theta}=1-\lambda.
\]
This is strictly negative if and only if $\lambda>1$. Thus, in the corner
regime, capability reduces deployment whenever breach-loss magnitude is high
enough.

\medskip
\noindent\textbf{Interior regime ($\lambda\theta>1$).}
From Proposition~\ref{prop:alpha_star},
\[
\alpha^*(\theta)=\theta+\mu+1-2\sqrt{\lambda\theta},
\]
so
\[
\frac{\partial \alpha^*}{\partial \theta}
=
1-\frac{\partial}{\partial\theta}\bigl(2\sqrt{\lambda\theta}\bigr)
=
1-\sqrt{\frac{\lambda}{\theta}}.
\]
This derivative is negative if and only if $\theta<\lambda$, zero at
$\theta=\lambda$, and positive for $\theta>\lambda$. The turning point
$\theta=\lambda$ is therefore economically meaningful: below it, the security
burden created by additional capability dominates the productivity gain;
above it, the productivity channel dominates again.

Now suppose $\lambda>1$. The corner paradox region is defined by
$\lambda\theta\le 1$ together with $1-\lambda<0$, which holds on
$(0,1/\lambda]$. The interior paradox region is defined by
$\lambda\theta>1$ and $\theta<\lambda$, which holds on $(1/\lambda,\lambda)$.
Since $1/\lambda<\lambda$ when $\lambda>1$, these two intervals adjoin,
yielding the full paradox region $(0,\lambda)$. At $\theta=\lambda$,
$\alpha^*$ attains its global minimum
\[
\alpha^*_{\min}=\mu+1-\lambda.
\]

If $\lambda\le 1$, then in the corner regime
$\partial \alpha^*/\partial\theta=1-\lambda\ge 0$, and in the interior regime
\[
\theta>\frac{1}{\lambda}\ge \lambda
\quad\Longrightarrow\quad
1-\sqrt{\frac{\lambda}{\theta}}\ge 0.
\]
So $\partial \alpha^*/\partial\theta\ge 0$ for all $\theta>0$, and no paradox
arises. The paradox therefore requires a sufficiently severe loss environment;
it is not a generic feature of better AI.
\end{proof}

\section{Proof of Proposition~\ref{prop:complementarity}: Effect of Governance Maturity on Deployment}
\label{app:proof_complementarity}

\begin{proof}
The proof follows by direct differentiation of $\alpha^*(\theta)$ from
Proposition~\ref{prop:alpha_star} with respect to $\lambda$ in each regime.
$\lambda$ is the model's index of inherited breach-loss magnitude, so reducing
$\lambda$ corresponds to closing the governance-capability gap through
improvements in containment quality, data segmentation, or access-control
architecture.

\medskip
\noindent\textbf{Corner regime ($\lambda\theta\le 1$).}
From $\alpha^*=\mu+\theta(1-\lambda)$,
\[
\frac{\partial\alpha^*}{\partial\lambda}=-\theta<0.
\]
Economically, in this regime the firm makes no AI-specific defensive
investment, so any reduction in $\lambda$ translates one-for-one into a
reduction of expected breach damage per unit of deployment. Each unit of
capability therefore gains $\theta$ units of deployment space from a unit
reduction in $\lambda$: higher-capability firms benefit more from governance
investment than lower-capability ones.

\medskip
\noindent\textbf{Interior regime ($\lambda\theta>1$).}
From $\alpha^*=\theta+\mu+1-2\sqrt{\lambda\theta}$,
\[
\frac{\partial\alpha^*}{\partial\lambda}
=-2\cdot\frac{1}{2}\sqrt{\frac{\theta}{\lambda}}
=-\sqrt{\frac{\theta}{\lambda}}<0.
\]
Once AI-specific defense becomes worthwhile, the marginal effect of
governance investment is attenuated by the firm's own defensive response:
reducing $\lambda$ also reduces the per-unit defensive burden
$\sqrt{\lambda\theta}-1$, so part of the gain is absorbed on the defense
margin rather than showing up directly in deployment. That is why the
interior derivative takes the $\sqrt{\theta/\lambda}$ form rather than the
linear $-\theta$ form of the corner regime.

\medskip
\noindent\textbf{Continuity at $\theta=1/\lambda$.}
At the regime boundary, the corner derivative equals $-\theta=-1/\lambda$,
while the interior derivative equals
$-\sqrt{\theta/\lambda}=-\sqrt{(1/\lambda)/\lambda}=-1/\lambda$. The two
expressions coincide, so $\partial\alpha^*/\partial\lambda$ is continuous
across regimes.

\medskip
\noindent\textbf{Overall implication.}
In both regimes, $\partial\alpha^*/\partial\lambda<0$, so closing the
governance-capability gap raises optimal deployment. The magnitude depends
on the current regime---proportional to $\theta$ in the corner regime, and
to $\sqrt{\theta/\lambda}$ in the interior regime---but the sign is
unambiguous. Governance investment and AI deployment are therefore
complements, not substitutes: the firm captures a larger deployment
response from governance investment precisely when its capability is high,
which is also when the security discount is largest
(Proposition~\ref{prop:discount}).
\end{proof}

\section{Proof of Proposition~\ref{prop:endogenous}: Endogenous Capability Choice}
\label{app:proof_endogenous}

\begin{proof}
The proof characterizes the global shape of the firm's value function
$V(\theta_c)$ and then compares the two endpoints $\theta_L$ and $\theta_F$.
The key object is the mapping from capability choice into optimal deployment,
because firm value is pinned down by the deployment rule.

Let $\theta_c\in\{\theta_L,\theta_F\}$ denote the firm's chosen capability.
From Lemma~\ref{lem:reduced}, the reduced-form profit in each regime takes the
form
\[
\pi(\alpha;\theta_c)=C(\theta_c)\alpha-\frac{\alpha^2}{2},
\]
where $C(\theta_c)$ is the regime-specific linear coefficient. Since the
optimal deployment is $\alpha^*(\theta_c)=C(\theta_c)$, substituting back
yields
\[
V(\theta_c)\equiv \pi^*(\theta_c)
= C(\theta_c)^2-\frac{C(\theta_c)^2}{2}
= \frac{[\alpha^*(\theta_c)]^2}{2}.
\]
Thus, capability affects firm value only through its effect on the optimal
deployment scale.

Differentiating with respect to $\theta_c$ gives
\[
\frac{dV}{d\theta_c}
=
\alpha^*(\theta_c)\frac{d\alpha^*}{d\theta_c}.
\]
Under Assumption~\ref{assump:interior}, $\alpha^*(\theta_c)>0$, so the sign of
$dV/d\theta_c$ is the same as the sign of $d\alpha^*/d\theta_c$. The value
comparison therefore reduces to the comparative static for deployment.

If $\lambda\le 1$, Proposition~\ref{prop:paradox} gives
$d\alpha^*/d\theta_c\ge 0$ for all $\theta_c>0$. Hence $V(\theta_c)$ is
weakly increasing, and the firm weakly prefers the frontier system. In low-gap
environments, stronger capability therefore preserves the standard monotone
adoption logic.

If $\lambda>1$, Proposition~\ref{prop:paradox} implies that
$d\alpha^*/d\theta_c<0$ for $\theta_c<\lambda$, equals zero at
$\theta_c=\lambda$, and is positive for $\theta_c>\lambda$. Therefore
$V(\theta_c)$ is strictly decreasing on $(0,\lambda)$ and strictly increasing
on $(\lambda,\infty)$, with a unique global minimum at $\theta_c=\lambda$.
The maximizer over any interval $[\theta_L,\theta_F]$ must therefore occur at
an endpoint. This is the sense in which weak governance can make higher
capability locally value-reducing before the value function recovers.

When both $\theta_L$ and $\theta_F$ lie in the interior regime
($\lambda\theta_L>1$), the firm upgrades if and only if
$V(\theta_F)>V(\theta_L)$, equivalently
$\alpha^*(\theta_F)>\alpha^*(\theta_L)$. Using
\[
\alpha^*(\theta)=\theta+\mu+1-2\sqrt{\lambda\theta},
\]
this condition becomes
\[
\theta_F-2\sqrt{\lambda\theta_F}
>
\theta_L-2\sqrt{\lambda\theta_L}.
\]
Rearranging,
\[
\theta_F-\theta_L
>
2\sqrt{\lambda}\bigl(\sqrt{\theta_F}-\sqrt{\theta_L}\bigr).
\]
Using
\[
\theta_F-\theta_L
=
(\sqrt{\theta_F}-\sqrt{\theta_L})(\sqrt{\theta_F}+\sqrt{\theta_L}),
\]
and dividing by $\sqrt{\theta_F}-\sqrt{\theta_L}>0$, we obtain
\[
\sqrt{\theta_F}+\sqrt{\theta_L}>2\sqrt{\lambda}.
\]
Equivalently, given $\theta_F>\theta_L$,
\[
\theta_F>
\left[\max\{0,\,2\sqrt{\lambda}-\sqrt{\theta_L}\}\right]^2.
\]
Thus the effective threshold over the feasible set $\theta_F>\theta_L$ is
\[
\bar{\theta}_F
=
\max\left\{
\theta_L,\,
\left[\max\{0,\,2\sqrt{\lambda}-\sqrt{\theta_L}\}\right]^2
\right\}.
\]
When $2\sqrt{\lambda}\le \sqrt{\theta_L}$, the second term is zero, so the
upgrade condition is automatically satisfied for any feasible
$\theta_F>\theta_L$. This threshold identifies how far the frontier system must
move beyond the relevant loss environment before the firm prefers the upgrade
again.

When the legacy system lies in the corner regime ($\lambda\theta_L\le 1$),
its value is
\[
V(\theta_L)=\frac{[\mu+\theta_L(1-\lambda)]^2}{2},
\]
and the upgrade condition $V(\theta_F)>V(\theta_L)$ must be evaluated
directly.
\end{proof}

\section{Proof of Proposition~\ref{prop:social_A}: First-Best Social Optimum}
\label{app:proof_social_A}

\begin{proof}
This proof reuses the private problem almost entirely. Once the social planner
internalizes the externality, the firm's loss environment is scaled from
$\lambda$ to $(1+e)\lambda$, so the first-best problem is a direct analogue of
the private benchmark with that substitution.

The social planner jointly maximizes
\[
W=\pi-e\cdot p\cdot L
\]
over $(\alpha,d)$. Since the firm's objective $\pi$ already includes the
private loss term $-p\cdot L$, the planner's objective can be written as
\[
W
=
(\theta+\mu)\alpha-\frac{\alpha^2}{2}
-(1+e)\frac{\lambda\alpha^2\theta}{\alpha+d}
-d.
\]
This is structurally identical to the firm's problem in \eqref{eq:profit},
except that $\lambda$ is replaced everywhere by $(1+e)\lambda$. Thus, the
first-best planner behaves as if the effective loss environment were more
severe than the firm perceives privately.

Applying Proposition~\ref{prop:alpha_star} with this substitution yields
\[
\alpha^{**}_{FB}(\theta)
=
\max\!\left\{0,\;
\begin{cases}
\mu+\theta\bigl(1-(1+e)\lambda\bigr), & (1+e)\lambda\theta\le 1,\\[4pt]
\theta+\mu+1-2\sqrt{(1+e)\lambda\theta}, & (1+e)\lambda\theta>1.
\end{cases}
\right\}.
\]

The $\max\{0,\cdot\}$ operator is needed because internalizing the externality
can push the planner to the corner solution $\alpha^{**}_{FB}=0$ even when the
private optimum remains positive. Socially optimal deployment is therefore
weakly below private deployment even before comparing regime boundaries.
\end{proof}

\section{Proof of Proposition~\ref{prop:social_B}: Second-Best Social Optimum}
\label{app:proof_social_B}

\begin{proof}
Unlike the first-best planner, the second-best regulator controls deployment
$\alpha$ but does not control the defense margin directly.
The firm therefore continues to choose defense according to
Lemma~\ref{lem:dstar}, and the regulator takes that private response as given
when choosing deployment.
Substituting the firm's privately optimal defense choice $d^*(\alpha)$ into
social welfare yields
\[
W_{SB}
=
\pi(\alpha,d^*)-e\cdot p(\alpha,d^*)\cdot L(\alpha,\theta).
\]

\medskip
\noindent\textbf{Corner regime ($\lambda\theta\le 1$).}
From Lemmas~\ref{lem:dstar}--\ref{lem:pstar}, we have $d^*=0$ and $p^*=1$.
Hence
\[
W_{SB}
=
(\theta+\mu)\alpha-\frac{\alpha^2}{2}-(1+e)\lambda\alpha\theta.
\]
This is a concave quadratic in $\alpha$, so the first-order condition yields
\[
\alpha^{**}_{SB}
=
\mu+\theta\bigl(1-(1+e)\lambda\bigr).
\]

\medskip
\noindent\textbf{Interior regime ($\lambda\theta>1$).}
From Lemmas~\ref{lem:dstar}--\ref{lem:pstar},
\[
d^*=\alpha(\sqrt{\lambda\theta}-1),
\qquad
p^*\cdot L=\alpha\sqrt{\lambda\theta}.
\]
Substituting into $W_{SB}$ gives
\begin{align*}
W_{SB}
&=
(\theta+\mu)\alpha-\frac{\alpha^2}{2}
-\alpha\sqrt{\lambda\theta}
-\alpha(\sqrt{\lambda\theta}-1)
-e\,\alpha\sqrt{\lambda\theta}\\
&=
\bigl[\theta+\mu+1-(2+e)\sqrt{\lambda\theta}\bigr]\alpha
-\frac{\alpha^2}{2}.
\end{align*}
Again this is a concave quadratic in $\alpha$, so the first-order condition
yields
\[
\alpha^{**}_{SB}
=
\theta+\mu+1-(2+e)\sqrt{\lambda\theta}.
\]

\medskip
Intuitively, once the regulator is limited to deployment control alone, it must
use $\alpha$ to manage both the firm's private security burden and the external
harm that the firm does not internalize.
In the corner regime, this scales up the private loss term $\lambda\theta$ by
the factor $(1+e)$.
In the interior regime, the firm's private defense choice already implies a
reduced-form security deduction of $(2\sqrt{\lambda\theta}-1)\alpha$, and the
externality adds a further $e\sqrt{\lambda\theta}\alpha$.
That is why the reduced-form linear coefficient becomes
$\theta+\mu+1-(2+e)\sqrt{\lambda\theta}$.

Combining the two regimes and imposing non-negativity gives
\[
\alpha^{**}_{SB}(\theta)
=
\max\!\left\{0,\;
\begin{cases}
\mu+\theta\bigl(1-(1+e)\lambda\bigr), & \lambda\theta\le 1,\\[4pt]
\theta+\mu+1-(2+e)\sqrt{\lambda\theta}, & \lambda\theta>1.
\end{cases}
\right\}.
\]

The regime threshold remains $\lambda\theta=1$ because the regulator controls
deployment but not the firm's defense rule.
Regulation can change the level of deployment, but it cannot change the
condition under which the firm finds defensive expenditure worthwhile.
This is why deployment regulation is blunter than joint control.
\end{proof}

\section{Proof of Proposition~\ref{prop:externality_expansion}: Externalities and the Social Paradox Region}
\label{app:proof_externality}

\begin{proof}
Let $x\equiv\lambda\theta$ and $E\equiv 1+e>1$. Define the private security
discount function
\[
H(x)=
\begin{cases}
x, & x\le 1,\\[3pt]
2\sqrt{x}-1, & x>1.
\end{cases}
\]
Then the private optimum can be written as
\[
\alpha^*(\theta)=\mu+\theta-H(x),
\]
where Assumption~\ref{assump:interior} ensures that this expression is
strictly positive.

The first-best planner internalizes the full social loss. Its problem is
equivalent to the private problem with $\lambda$ replaced by $E\lambda$.
Therefore,
\[
\alpha^{**}_{FB}(\theta)
=
\max\{0,\mu+\theta-H(Ex)\}.
\]

The second-best regulator controls deployment but not the firm's defense
choice. The firm therefore chooses defense using the private threshold $x=1$.
Define the second-best social security discount as
\[
S_e(x)=
\begin{cases}
Ex, & x\le 1,\\[3pt]
(2+e)\sqrt{x}-1, & x>1.
\end{cases}
\]
Then
\[
\alpha^{**}_{SB}(\theta)
=
\max\{0,\mu+\theta-S_e(x)\}.
\]

\medskip
\noindent\textbf{Part 1: Ordering of private, first-best, and second-best deployment.}
Because $H(\cdot)$ is increasing and $E>1$, we have
\[
H(Ex)\ge H(x).
\]
It follows that
\[
\max\{0,\mu+\theta-H(Ex)\}\le \mu+\theta-H(x),
\]
so
\[
\alpha^{**}_{FB}(\theta)\le \alpha^*(\theta).
\]

It remains to show that $\alpha^{**}_{SB}(\theta)\le\alpha^{**}_{FB}(\theta)$.
It is sufficient to show that
\[
S_e(x)\ge H(Ex)
\]
for all $x>0$.

First suppose $x\le 1$. Then $S_e(x)=Ex$. If $Ex\le 1$, then
$H(Ex)=Ex$, so the two are equal. If $Ex>1$, then
\[
H(Ex)=2\sqrt{Ex}-1.
\]
Let $y=Ex>1$. Since
\[
y-(2\sqrt{y}-1)=(\sqrt{y}-1)^2\ge 0,
\]
we have $Ex\ge H(Ex)$.

Now suppose $x>1$. Then
\[
S_e(x)=(2+e)\sqrt{x}-1=(E+1)\sqrt{x}-1,
\]
while
\[
H(Ex)=2\sqrt{Ex}-1=2\sqrt{E}\sqrt{x}-1.
\]
Hence
\[
S_e(x)-H(Ex)
=
\bigl(E+1-2\sqrt{E}\bigr)\sqrt{x}
=
(\sqrt{E}-1)^2\sqrt{x}\ge 0.
\]
Thus $S_e(x)\ge H(Ex)$ for all $x>0$, which implies
\[
\alpha^{**}_{SB}(\theta)\le\alpha^{**}_{FB}(\theta).
\]
Together, these inequalities prove
\[
\alpha^{**}_{SB}(\theta)\le
\alpha^{**}_{FB}(\theta)\le
\alpha^*(\theta).
\]
The inequalities are weak because the max operator can bind. For example, if
both unconstrained social optima are negative, then
$\alpha^{**}_{SB}=\alpha^{**}_{FB}=0$. When the relevant social optima are
interior, the inequalities become strict whenever the corresponding discount
inequality is strict.

\medskip
\noindent\textbf{Part 2: First-best paradox conditions.}
The first-best optimum is the private optimum with $\lambda$ replaced by
$E\lambda$. In the first-best corner regime, $E\lambda\theta\le 1$, and
\[
\alpha^{**}_{FB}=\mu+\theta(1-E\lambda).
\]
Therefore,
\[
\frac{\partial \alpha^{**}_{FB}}{\partial \theta}
=
1-E\lambda.
\]
The first-best corner-regime paradox condition is
\[
E\lambda>1,
\]
or equivalently
\[
(1+e)\lambda>1.
\]

In the first-best interior regime, $E\lambda\theta>1$, and
\[
\alpha^{**}_{FB}
=
\theta+\mu+1-2\sqrt{E\lambda\theta}.
\]
Thus
\[
\frac{\partial \alpha^{**}_{FB}}{\partial \theta}
=
1-\sqrt{\frac{E\lambda}{\theta}}.
\]
The first-best interior-regime paradox condition is therefore
\[
\theta<E\lambda=(1+e)\lambda.
\]

\medskip
\noindent\textbf{Part 3: Second-best paradox conditions.}
In the second-best case, the firm still chooses defense using the private
threshold $\lambda\theta=1$. In the corner regime, $\lambda\theta\le 1$, and
\[
\alpha^{**}_{SB}
=
\mu+\theta(1-E\lambda).
\]
Therefore,
\[
\frac{\partial \alpha^{**}_{SB}}{\partial \theta}
=
1-E\lambda.
\]
The second-best corner-regime paradox condition is
\[
(1+e)\lambda>1.
\]

In the second-best interior regime, $\lambda\theta>1$, and
\[
\alpha^{**}_{SB}
=
\theta+\mu+1-(2+e)\sqrt{\lambda\theta}.
\]
Thus
\[
\frac{\partial \alpha^{**}_{SB}}{\partial \theta}
=
1-\frac{2+e}{2}\sqrt{\frac{\lambda}{\theta}}.
\]
The second-best interior-regime paradox condition is therefore
\[
\theta<
\left(\frac{2+e}{2}\right)^2\lambda.
\]

\medskip
\noindent\textbf{Part 4: Expansion relative to the private paradox region.}
The private interior paradox condition is $\theta<\lambda$. The first-best
interior threshold is $(1+e)\lambda$, and the second-best interior threshold is
\[
\left(\frac{2+e}{2}\right)^2\lambda.
\]
Because $e>0$,
\[
(1+e)\lambda>\lambda
\]
and
\[
\left(\frac{2+e}{2}\right)^2\lambda>\lambda.
\]
Thus the social paradox regions are weakly wider than the private paradox
region.

There also exist parameter values such that
\[
\theta>\lambda
\quad\text{but}\quad
\theta<(1+e)\lambda,
\]
or
\[
\theta>\lambda
\quad\text{but}\quad
\theta<
\left(\frac{2+e}{2}\right)^2\lambda.
\]
For such values, the firm is outside the private paradox region but inside the
social paradox region. This proves that externalities can convert a privately
normal capability response into a socially paradoxical one.
\end{proof}

\section{Interpretation of the Damage Channel}
\label{app:damage_channel}

The baseline $L=\lambda\alpha\theta$ reflects the gap: $\theta$ co-varies with
authority exposure because governance has not constrained the combination. See
Assumption~\ref{assump:bundling} for the formal statement. $\mu$-type organizational readiness investments raise productivity without expanding the
damage channel, because they do not co-vary with authority exposure under gap
conditions. Appendix~\ref{app:mu_damage} shows the paradox is robust to partial inclusion of $\mu$ in the damage
channel.

\section{Complementary Organizational Readiness in the Damage Channel}
\label{app:mu_damage}

The baseline model assumes that complementary organizational readiness $\mu$
raises productivity but does not expand conditional breach damage. This
section relaxes that separation. Suppose
\[
L(\alpha,\theta,\mu)=\lambda\alpha(\theta+\omega\mu),
\qquad \omega\in[0,1].
\]
The parameter $\omega$ captures the extent to which readiness-enhancing
investments also expand operational exposure.

The optimal deployment expression becomes
\[
\alpha^*_\omega =
\max\!\left\{0,\;
\begin{cases}
    \mu+\theta-\lambda(\theta+\omega\mu),
    & \lambda(\theta+\omega\mu)\le 1,\\[4pt]
    \theta+\mu+1-2\sqrt{\lambda(\theta+\omega\mu)},
    & \lambda(\theta+\omega\mu)>1.
\end{cases}
\right\}.
\]
In the interior regime,
\[
\frac{\partial\alpha^*_\omega}{\partial\theta}
=
1-\sqrt{\frac{\lambda}{\theta+\omega\mu}}.
\]
Thus the paradox condition becomes
\[
\theta+\omega\mu<\lambda.
\]
The baseline case $\omega=0$ gives $\theta<\lambda$. Allowing $\mu$ to enter
the damage channel shifts the boundary, but it does not remove the mechanism.
For any given $(\theta,\lambda,\mu)$ with $\theta<\lambda$, the paradox
continues to hold for
\[
\omega<\frac{\lambda-\theta}{\mu}.
\]
Thus partial exposure spillovers from organizational readiness narrow the
paradox region but do not overturn the result whenever the spillover is
sufficiently limited.

\section{Generalized Authority Exposure and Breach Damage}
\label{app:generalized_L}

The baseline model uses the normalization $a(\theta)=\theta$. This appendix
relaxes that normalization by setting
\[
a(\theta)=\theta^\gamma,\qquad \gamma\ge 0,
\]
so that conditional breach damage becomes
\[
L(\alpha,\theta)=\lambda\alpha a(\theta)
=
\lambda\alpha\theta^\gamma.
\]
The parameter $\gamma$ captures the elasticity of authority exposure, and
therefore conditional breach damage, with respect to capability. The baseline
model sets $\gamma=1$.

Given this damage function, optimal deployment becomes
\begin{equation}
    \alpha^*_\gamma =
    \max\!\left\{0,\;
    \begin{cases}
        \mu+\theta-\lambda\theta^\gamma,
        & \lambda\theta^\gamma\le 1, \\[4pt]
        \theta+\mu+1-2\sqrt{\lambda\theta^\gamma},
        & \lambda\theta^\gamma>1.
    \end{cases}
    \right\}.
    \label{eq:alpha_gamma}
\end{equation}

In the corner regime, the derivative with respect to capability is
\[
\frac{\partial\alpha^*_\gamma}{\partial\theta}
=
1-\lambda\gamma\theta^{\gamma-1}.
\]
Thus sign reversal occurs in the corner regime when
\[
\lambda\gamma\theta^{\gamma-1}>1.
\]

In the interior regime, the derivative is
\[
\frac{\partial\alpha^*_\gamma}{\partial\theta}
=
1-\gamma\sqrt{\lambda}\,\theta^{\gamma/2-1}.
\]
Thus sign reversal occurs in the interior regime when
\[
\gamma\sqrt{\lambda}\,\theta^{\gamma/2-1}>1.
\]
Equivalently, for $\gamma\ne 2$, this condition can be written as
\[
\theta^{2-\gamma}<\gamma^2\lambda,
\]
with the direction interpreted through the preceding derivative condition.
For $\gamma=2$, the interior derivative is
\[
1-2\sqrt{\lambda},
\]
so sign reversal occurs when $\lambda>1/4$.

These expressions show that the paradox does not depend on treating authority
exposure and capability as identical. It requires a sufficiently strong
capability-linked exposure channel over the relevant parameter region. When
$\gamma=0$, capability has no effect on authority exposure or conditional
breach damage. The derivative in both regimes is then nonnegative, and the
deployment paradox disappears.

\section{Generalized Breach Probability and the Attack-Surface Channel}
\label{app:generalized_p}

We generalize the breach-probability function to
\[
p(\alpha,d)=\frac{\alpha^\beta}{\alpha^\beta+d},\qquad \beta\ge 1.
\]
The baseline case is $\beta=1$. This appendix studies how superlinear
attack-surface growth changes local deployment incentives under interior
feasibility. The goal is not to establish a global theorem for all probability
forms, but to show that once breach probability grows more than proportionally
with deployment, the attack-surface channel can reinforce the paradox.

Under this specification, expected security loss becomes
\[
p(\alpha,d)\cdot L(\alpha,\theta)
=
\frac{\alpha^\beta}{\alpha^\beta+d}\cdot \lambda\alpha\theta
=
\frac{\lambda\alpha^{\beta+1}\theta}{\alpha^\beta+d}.
\]

The firm's objective is therefore
\[
\pi(\alpha,d)
=
(\theta+\mu)\alpha-\frac{\alpha^2}{2}
-\frac{\lambda\alpha^{\beta+1}\theta}{\alpha^\beta+d}
-d.
\]

For a given $\alpha$, differentiate with respect to $d$:
\[
\frac{\partial \pi}{\partial d}
=
\frac{\lambda\alpha^{\beta+1}\theta}{(\alpha^\beta+d)^2}-1.
\]
The interior first-order condition implies
\[
(\alpha^\beta+d)^2=\lambda\alpha^{\beta+1}\theta,
\]
so
\[
\alpha^\beta+d=\alpha^{(\beta+1)/2}\sqrt{\lambda\theta},
\]
and thus
\[
d^*_\beta(\alpha,\theta)
=
\alpha^{(\beta+1)/2}\sqrt{\lambda\theta}-\alpha^\beta.
\]

Substituting this expression into the generalized breach-probability function
yields
\[
p^*_\beta
=
\frac{\alpha^\beta}{\alpha^{(\beta+1)/2}\sqrt{\lambda\theta}}
=
\frac{\alpha^{(\beta-1)/2}}{\sqrt{\lambda\theta}}.
\]
For $\beta>1$, $p^*_\beta$ is increasing in $\alpha$: the firm can no longer
fully offset superlinear attack-surface growth through proportionate defense.

Substituting $d^*_\beta$ and $p^*_\beta$ back into profit gives the
reduced-form objective
\[
\pi^*_\beta(\alpha)
=
(\theta+\mu)\alpha-\frac{\alpha^2}{2}
-2\alpha^{(\beta+1)/2}\sqrt{\lambda\theta}
+\alpha^\beta.
\]

The first-order condition for deployment is
\[
F(\alpha,\theta)
\equiv
\theta+\mu-\alpha
-(\beta+1)\alpha^{(\beta-1)/2}\sqrt{\lambda\theta}
+\beta\alpha^{\beta-1}
=0.
\]

To study the local comparative static with respect to capability, apply the
implicit function theorem:
\[
\frac{d\alpha^*}{d\theta}
=
-\frac{F_\theta}{F_\alpha}.
\]
Under interior concavity, the sign of $d\alpha^*/d\theta$ is governed by the
sign of $F_\theta$, where
\[
F_\theta
=
1-\frac{\beta+1}{2}\alpha^{(\beta-1)/2}\sqrt{\frac{\lambda}{\theta}}.
\]
Hence a local sign-reversal condition is
\[
\frac{\beta+1}{2}\alpha^{(\beta-1)/2}\sqrt{\frac{\lambda}{\theta}} > 1,
\]
equivalently,
\[
\frac{(\beta+1)\alpha^{(\beta-1)/2}\sqrt{\lambda}}{2\sqrt{\theta}} > 1.
\]

This condition nests the baseline case. When $\beta=1$, it reduces to
\[
\sqrt{\frac{\lambda}{\theta}}>1
\quad\Longleftrightarrow\quad
\theta<\lambda,
\]
which is exactly the interior paradox condition in
Proposition~\ref{prop:paradox}. For $\beta>1$, the additional factor
$\alpha^{(\beta-1)/2}$ makes the inequality easier to satisfy as deployment
expands. The baseline $\beta=1$ should therefore be read as a conservative
benchmark: once governance failures make attack-surface growth superlinear, the
attack-surface channel can reinforce the paradox locally.

\section{Generalized Capability Productivity}
\label{app:generalized_prod}

We generalize the productivity term from $\theta\alpha$ to
$\theta^\eta\alpha$, where $\eta>0$. The firm's objective becomes
\[
\pi(\alpha,d)
=
(\theta^\eta+\mu)\alpha-\frac{\alpha^2}{2}
-\frac{\alpha}{\alpha+d}\lambda\alpha\theta
-d.
\]
The damage channel remains $L=\lambda\alpha\theta$, so this extension varies
only the productivity side of the model.

The optimal deployment level is
\[
\alpha^*_\eta =
\max\!\left\{0,\;
\begin{cases}
    \mu+\theta^\eta-\lambda\theta,
    & \lambda\theta\le 1,\\[4pt]
    \theta^\eta+\mu+1-2\sqrt{\lambda\theta},
    & \lambda\theta>1.
\end{cases}
\right\}.
\]

In the corner regime,
\[
\frac{\partial\alpha^*_\eta}{\partial\theta}
=
\eta\theta^{\eta-1}-\lambda.
\]
Thus sign reversal occurs when
\[
\eta\theta^{\eta-1}<\lambda.
\]

In the interior regime,
\[
\frac{\partial\alpha^*_\eta}{\partial\theta}
=
\eta\theta^{\eta-1}
-
\sqrt{\lambda}\theta^{-1/2}.
\]
Thus sign reversal occurs when
\[
\eta\theta^{\eta-1/2}<\sqrt{\lambda}.
\]
At $\eta=1$, this reduces to
\[
\theta<\lambda,
\]
which is the baseline interior paradox condition.

This extension shows how the paradox depends on the relative curvature of the
productivity and damage channels. Stronger productivity scaling narrows the
paradox region. Stronger loss exposure widens it. The baseline model therefore
captures the transparent case in which both channels scale linearly with
capability before optimal defense is chosen.

\section{Comparative Statics}
\label{app:comp_statics}

\begin{table}[htbp]
\centering
\caption{Comparative Statics (Interior Solution)}
\label{tab:comp_statics}
\begin{tabular}{lcccc}
\toprule
& $\alpha^*$ & $d^*$ & $p^*$ & $\pi^*$ \\
\midrule
$\uparrow\theta$ & $+$ if $\theta>\lambda$; $-$ if $\theta<\lambda$
    & ambiguous & $-$ & same sign as $\partial\alpha^*/\partial\theta$ \\
$\uparrow\lambda$ & $-$ & per-unit $\uparrow$, total ambiguous & $-$ & $-$ \\
$\uparrow\mu$     & $+$ & $+$ & 0 & $+$ \\
\bottomrule
\end{tabular}
\\\smallskip
\small $\pi^*=(\alpha^*)^2/2$ at the interior optimum.
\end{table}

\begin{figure}[H]
    \centering
    \includegraphics[width=0.68\textwidth]{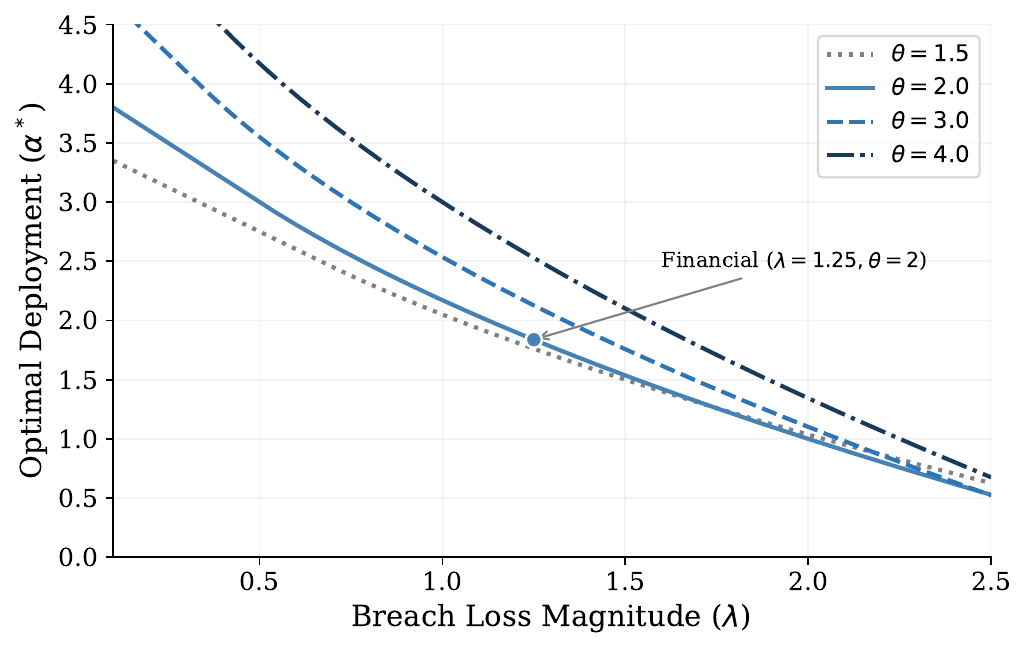}
    \caption{Governance maturity and AI deployment
    (Proposition~\ref{prop:complementarity}). $\alpha^*$ as a function of
    $\lambda$ for $\theta\in\{1.5,2,3,4\}$, $\mu=2$. All curves monotonically
    decreasing: lower $\lambda$ enables deeper deployment.}
    \label{fig:complement}
\end{figure}

\begin{figure}[H]
    \centering
    \includegraphics[width=0.68\textwidth]{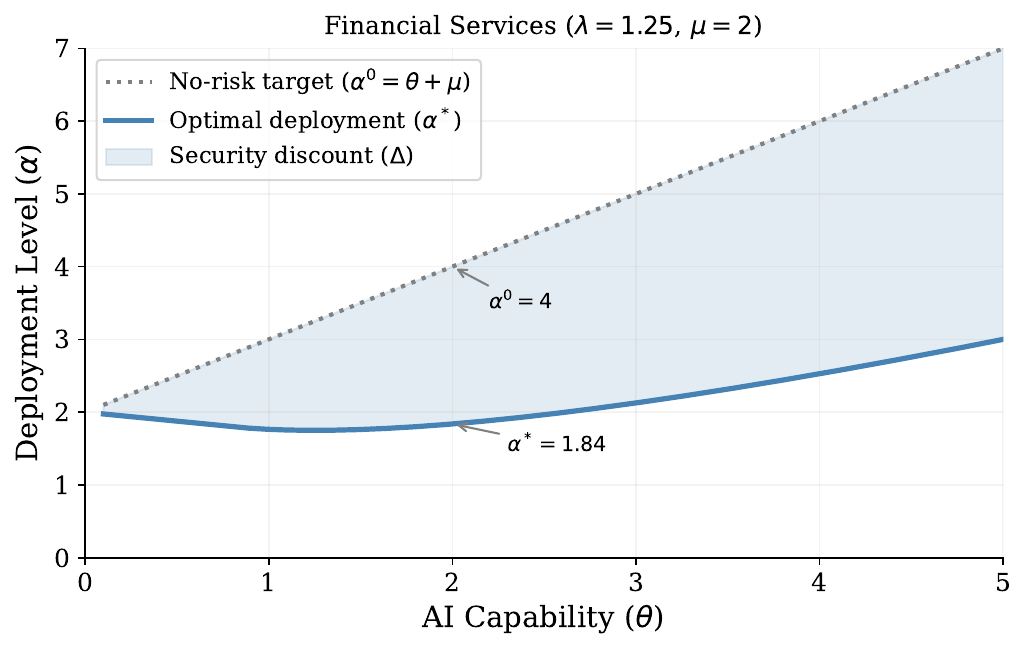}
    \caption{The security discount ($\lambda=1.25$, $\mu=2$). Shaded area:
    deployment foregone due to the gap. At $\theta=2$: $\alpha^*=1.84$ vs.\
    $\alpha^0=4$, a 54\% gap-induced security discount.}
    \label{fig:discount_app}
\end{figure}

\begin{figure}[H]
    \centering
    \includegraphics[width=0.68\textwidth]{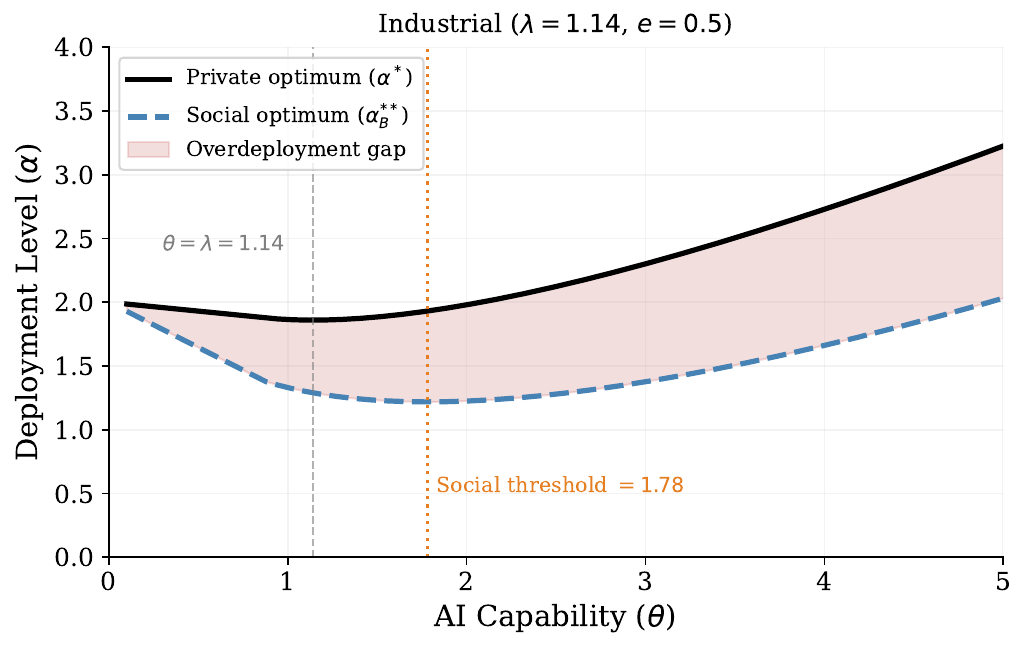}
    \caption{Externality expands the paradox ($\lambda=1.14$, $e=0.5$,
    $\mu=2$). Solid: $\alpha^*$; dashed: $\alpha^{**}_{SB}$. Private paradox
    threshold $\theta=1.14$ (gray dashed); social paradox threshold
    $\theta=1.78$ (orange dotted).}
    \label{fig:social_app}
\end{figure}

\section{Extensions}
\label{app:extensions}

\subsection{Governance Adjustment and the Paradox Region}
\label{app:gov_adjustment}

When capability improvements are not matched by governance investment, the
paradox region persists over a wider range of conditions. By contrast, when
governance investment reduces $\lambda$, the paradox region narrows and the
negative capability--deployment relationship weakens. Industries that improve
governance maturity alongside capability adoption should therefore experience
an earlier exit from the paradox region.

\begin{remark}[Governance Adjustment and Boundary Conditions]
\label{remark:dynamic}
Consider an extension in which capability improvement is accompanied by
governance investment that reduces $\lambda$. The paradox region is defined by
$\theta < \lambda$. Any force that raises $\theta$ or lowers $\lambda$ moves
the firm toward the standard monotone relationship. Three cases apply:
(1)~\textbf{Governance-led adjustment}: paradox region shrinks because
governance investment lowers $\lambda$. (2)~\textbf{Capability-led
adjustment}: paradox region shrinks because capability rises beyond the
threshold. (3)~\textbf{Limited governance adjustment}: if $\lambda$ remains
high, the paradox region remains broader. A richer dynamic optimization is
beyond the scope of the current paper.
\end{remark}

\subsection{Endogenous Governance and Rational Security Debt}
\label{app:endogenous_gov}

A natural critique of the baseline model is that firms suffering from the
deployment paradox should immediately invest to eliminate the governance gap.
We show here that when governance restructuring entails convex organizational
costs, a residual gap may optimally persist.

Suppose the firm inherits $\lambda_0 > 1$ and can invest $I \in [0, \lambda_0]$
to achieve $\lambda(I) = \lambda_0 - I$, at convex cost $C(I) = kI^2/2$, where 
$k > 0$ captures institutional friction. The firm maximizes $W(I) = V(\lambda_0 - I) - kI^2/2$, 
where $V(\lambda)$ denotes the maximized baseline profit. At an interior optimum,
the FOC is $-V'(\lambda(I^*)) = kI^*$. By the envelope theorem applied to the 
interior-regime profit $\pi(\alpha, d; \lambda)$, the marginal value of reducing $\lambda$ is:
\begin{equation}
    -V'(\lambda) = -\left.\frac{\partial\pi(\alpha^*, d^*;\,\lambda)}{\partial\lambda}\right|
        = \frac{(\alpha^*)^2\theta}{\alpha^* + d^*}.
    \label{eq:gov_mb}
\end{equation}
Substituting the optimal defense $d^* = \alpha^*(\sqrt{\lambda\theta}-1)$
gives $\alpha^* + d^* = \alpha^*\sqrt{\lambda\theta}$, which simplifies
\eqref{eq:gov_mb} to:
\begin{equation}
    -V'(\lambda) = \alpha^*(\lambda)\sqrt{\theta/\lambda} > 0.
    \label{eq:gov_mb2}
\end{equation}
This holds in the interior regime ($\lambda\theta > 1$); in the corner
regime ($\lambda\theta \le 1$), $d^*=0$ and an analogous calculation yields
$-V'(\lambda) = \theta\cdot\alpha^* > 0$, confirming the marginal benefit of
reducing $\lambda$ is strictly positive across both regimes.

Whenever the solution is interior, the optimal governance investment satisfies:
\begin{equation}
    k I^* =
    \begin{cases}
        \alpha^*(\lambda^*)\sqrt{\theta/\lambda^*}, & \lambda^*\theta > 1, \\[4pt]
        \theta\cdot\alpha^*(\lambda^*),             & \lambda^*\theta \le 1,
    \end{cases}
    \qquad \lambda^* \equiv \lambda_0 - I^*.
    \label{eq:gov_foc}
\end{equation}

\begin{remark}[Rational Accumulation of Security Debt]
\label{remark:security_debt}
As $k \to \infty$, $I^* \to 0$ and $\lambda^* \to \lambda_0$. By continuity,
if the inherited environment initially lies in the paradox region---that is,
if $\lambda_0 > \max\{1, \theta\}$---there exists a threshold $\bar{k}$ such 
that $\lambda^* > \max\{1, \theta\}$ for all $k > \bar{k}$: the firm remains 
in the paradox region despite having a positive incentive to reduce $\lambda$. 
The persistence of the governance-capability gap is therefore a structural 
consequence of convex organizational friction, not managerial myopia. From 
\eqref{eq:gov_foc}, external interventions that reduce adjustment frictions or require minimum
governance investment---such as governance subsidies or segmentation
standards---can raise effective governance investment and endogenously contract
the paradox region, consistent with the $\lambda$-reducing policy channel
discussed in Section~\ref{sec:discussion}.
\end{remark}

\subsection{Strategic Attacker}
\label{app:strategic_attacker}

If attacker capability scales with $\theta$, breach risk becomes more sensitive
to AI deployment and capability, further tightening the paradox region. The 
baseline likely understates the paradox in such environments, but a 
game-theoretic treatment is left for future work.

\subsection{Heterogeneous AI Applications}
\label{app:heterogeneous_apps}

With high-authority ($\alpha_H$) and low-authority ($\alpha_L$) AI, high-gap
firms would shift portfolios toward lower-authority applications as $\theta$
increases within the paradox region---substituting toward AI uses outside the
gap even as aggregate deployment declines.

\end{document}